\documentclass[11pt]{article}

\usepackage[utf8x]{inputenc}

\usepackage{amssymb}
\usepackage{amsmath}
\usepackage{verbatim}

\usepackage{graphicx}
\usepackage{color}

\usepackage{url}

\usepackage[algoruled]{algorithm2e}
\SetKwInput{Input}{Input}
\SetKwInput{Output}{Output}

\title{}
\author{}

\global\long\def\F{F}

\global\long\def\N{\mathbb{N}}
\global\long\def\Z{\mathbb{Z}}

\DeclareMathOperator{\polylog}{polylog}

\DeclareMathOperator{\corner}{C}

\def\lpopt{LP_{OPT}}

\def\UFPP{UFPP}
\def\area{\mbox{area}}
\def\PP{\mathcal{P}}
\def\ct{corner}
\def\UFPPG{$\mbox{\sc ufpp}(G)$}
\def\ISR{ITS}

\def\III{\mathcal{I}}
\def\hp{high profit}

\def\bs{\backslash}

   \usepackage{a4wide}

   \newtheorem{theorem}{Theorem}[section]
   \newtheorem{lemma}[theorem]{Lemma}
   \newtheorem{proposition}[theorem]{Proposition}
   \newtheorem{corollary}[theorem]{Corollary}
   \newtheorem{definition}[theorem]{Definition}

 \newcommand{\QED}{$\Box$}  
 \newcommand{\pfbegin}{\noindent{\em Proof:}}

 \newenvironment{proof}{\vspace{1ex}\pfbegin\ }
 	{\hfill\QED
 	\medskip
 	}
 \newenvironment{pfsk}{\vspace{1ex}\noindent{\em Proof sketch:}\ }
  	{\hfill\QED
  	\medskip
  	}
 \newenvironment{pfof}[1]{\vspace{1ex}\noindent{\em Proof of #1:}\ }
  	{\hfill\QED
	\medskip  	
  	}

\numberwithin{equation}{section} 
\numberwithin{figure}{section} 
  \newtheorem{fact}[theorem]{Fact}

\title{A Constant Factor Approximation Algorithm for Unsplittable Flow on Paths}

\author{Paul Bonsma\thanks{Humboldt Universit\"{a}t zu Berlin, Computer Science Department,
Unter den Linden 6, 10099 Berlin, Germany, {\tt bonsma@informatik.hu-berlin.de}. The author was supported by DFG grant BO 3391/1-1.} \and Jens Schulz\thanks{Technische Universit\"at Berlin, Institute of Mathematics, Stra\ss e des 17.\ Juni 136, 10623 Berlin, Germany, {\tt jschulz@math.tu-berlin.de}. The author was supported by the DFG Research Center \textsc{Matheon} \textit{Mathematics for key technologies} in Berlin. }
\and Andreas Wiese\thanks{Universit\`{a} di Roma "La Sapienza", Via Ariosto 25, 00185 Rome, Italy, {\tt wiese@dis.uniroma1.it}.
The author was supported by the German Academic Exchange Service (DAAD).
}  }

\begin{document}

\maketitle

\begin{abstract}
In the unsplittable flow problem on a path, 
we are given a capacitated path $P$ and $n$
tasks, each task having a demand, a profit, and start and end vertices. The goal is
to compute a maximum profit set of tasks, such that for each edge $e$ of $P$, the total demand of selected tasks that use $e$ does not exceed the capacity of $e$. 
This is a well-studied problem that has been studied under alternative names, such as resource allocation, bandwidth allocation, resource constrained scheduling, temporal knapsack and interval packing.

We present a polynomial time constant-factor approximation algorithm for this problem.
This improves on the previous best known approximation ratio of $O(\log n)$.
The approximation ratio of our algorithm is $7+\epsilon$ for any $\epsilon>0$.

We introduce several novel algorithmic techniques, which might be of independent interest: a framework which reduces the problem to instances with a bounded range of  capacities, and a new geometrically inspired dynamic program which solves a special case of the maximum weight independent set of rectangles problem to optimality.
In the setting of resource augmentation, wherein the capacities can be slightly violated,
we give a~$(2+\epsilon)$-approximation algorithm.
In addition, we show that the problem is strongly NP-hard
even if all edge capacities are equal and all demands are either~1,~2, or~3.

\end{abstract}

\pagestyle{myheadings}
\thispagestyle{plain}
\markboth{Paul Bonsma, Jens Schulz, and Andreas Wiese}{A Constant Factor Approximation for Unsplittable Flow on Paths}

\section{Introduction}

In the Unsplittable Flow Problem on a Path (UFPP), we are given
a path $P=(V,E)$ with an integral capacity $u_{e}$ for each edge~$e\in E$. In addition, we
are given a set of~$n$ tasks~$T$ where each task~$i\in T$ is characterized
by a \emph{start vertex}~$s_{i}\in V$, an \emph{end vertex}~$t_{i}\in V$,
a \emph{demand}~$d_{i}\in\N$, and a \emph{profit}~$w_{i}\in\N$.
A task~$i$ {\em uses an edge $e\in E$} if $e$ lies on the path from $s_i$ to $t_i$.
The aim is to compute a set of tasks $F\subseteq T$ with maximum
total profit $\sum_{i\in F}w_{i}$ such that for each edge, 
the sum of the demands of all tasks in $F$ that use this edge does not exceed its capacity.

The name of this problem is motivated by an interpretation as a multicommodity flow problem, 
where each task corresponds to a commodity. 
The term {}``unsplittable'' 
means that the total amount of flow $d_{i}$ from each commodity
$i$ has to be routed completely along the path from the source $s_{i}$ to the sink $t_{i}$
or not at all. 
There are several settings and applications in which this problem occurs, 
and several other interpretations of the problem. 
Therefore, this problem, and close variants thereof, have been studied under the names 
{\em bandwidth allocation}~\cite{BYBC06,CHT2002,LMV2000}, 
{\em admission control}~\cite{PUW2000}, 
{\em interval packing}~\cite{CWMX-ESA2010}
{\em temporal knapsack}~\cite{BFHMTU05}, 
{\em multicommodity demand flow}~\cite{I-MCF-factor4-trees}, 
{\em unsplittable flow problem}~\cite{BCES2006,SODA-unsplit-flow,CCGK2007,CEKApprox2009}, {\em scheduling with fixed start and end times}~\cite{AS1987},
and {\em resource allocation}~\cite{BBFNS2000,CCKR2002,DPS2010,PUW2000}.
In many applications, the vertices correspond to time points, and tasks have fixed start and end times. Within this time interval they consume a given amount of a common resource, of which the available amount varies over time.

\UFPP\ is easily seen to be (weakly) NP-hard, since it contains the Knapsack problem as a special case (in case the path is just a single edge). 
In addition, Darmann et~al.~\cite{DPS2010} show that the special case where all profits and all capacities are uniform is also weakly NP-hard.
Chrobak et~al.~\cite{CWMX-ESA2010} strengthen this result by showing strong NP-hardness in this case. In addition, they show that the case where the profits equal the demands is strongly NP-hard.
These results show that the problem admits no polynomial time approximation scheme (PTAS) unless $P=NP$.
On the other hand, the special case of a single edge (Knapsack) admits an FPTAS. When the number of edges is bounded by a constant, \UFPP\ admits a PTAS since it is a special case of Multi-Dimensional Knapsack~\cite{FriezeClarke84}.

Most of the research on \UFPP\ has focused on two restricted cases: firstly, the special case in which all capacities are equal has been well-studied, which is also known as the {\em Resource Allocation Problem (RAP)}~\cite{BBFNS2000,CCKR2002,DPS2010,DS2006,PUW2000}.
A more general special case of \UFPP\ is given by the {\em No-Bottleneck Assumption (NBA)}: 
in that case it is required that~$\max_i d_{i}\leq \min_e u_{e}$ (this holds in particular for RAP). 
We will denote this restriction of the problem by {\em UFPP-NBA}.
For UFPP-NBA, a $(2+\epsilon)$-approximation algorithm is
known~\cite{I-MCF-factor4-trees}, which matches the earlier best approximation ratio for RAP~\cite{CCKR2002}.

Many previous papers on UFPP partition the tasks into small and large tasks, and use different algorithmic techniques for these two groups. 
For a task $i$, denote by $b(i)$ the minimum capacity among all edges used by task $i$.
For $\delta$ with $0 <\delta \le 1$, we say that a task~$i$ is {\em $\delta$-small} if $d_i \le \delta \cdot b(i)$ holds, 
and {\em $\delta$-large} otherwise.
The two main algorithmic techniques that have been used in previous results are dynamic programming (for large tasks)
and rounding of solutions to the linear programming relaxation of the problem (for small tasks). These techniques work well when the NBA holds.

However, there are several important obstacles 
that prevent these techniques to be generalized to the general case of \UFPP.
For example, Chakrabarti~et~al.~\cite{CCGK2007} show that 
under the NBA the natural LP-relaxation of \UFPP\ has a constant integrality gap. 
However, without this assumption the integrality gap can be as large as~$\Omega(n)$~\cite{CCGK2007}.  
Moreover, the NBA implies that if all tasks are $\delta$-large,
then in any solution there can be at most $2\left\lfloor 1/\delta^{2}\right\rfloor$
tasks which use each edge. 
This property is useful for setting up a dynamic program; see~\cite{CCGK2007} and Section~\ref{sec:medium-tasks}. 
Without the NBA this is no longer possible.

Despite these obstacles, there are a few breakthrough results for (general) \UFPP:
The best known polynomial time algorithm by Bansal et al.~\cite{SODA-unsplit-flow} achieves an approximation factor of $O(\log n)$,
thus beating the integrality gap of the natural LP-relaxation. 
This result has been generalized to trees by Chekuri et~al.~\cite{CEKApprox2009}. In addition, they gave a linear programming relaxation for UFPP with integrality gap $O(\log^2 n)$~\cite{CEKApprox2009}.
Finally, Bansal et al.~\cite{BCES2006} gave a $(1+\epsilon)$-approximation algorithm with
quasi-polynomial running time, which additionally requires that the
capacities and the demands are quasi-polynomial, i.e.\ bounded
by~$2^{\polylog n}$.
Nevertheless, it remained an open question whether UFPP admits a constant factor approximation algorithm (this was asked e.g. in~\cite{SODA-unsplit-flow,CEKApprox2009}).

\subsection{Our Contribution and Outline}
We present the first polynomial time constant-factor approximation algorithm for the general case of UFPP. The algorithm has an approximation ratio of~$7+\epsilon$, 
for arbitrary $\epsilon>0$.

To obtain this result we introduce several new algorithmic techniques, which are interesting in their own right.
We develop a useful viewpoint which allows us to reduce the problem to a special case of the maximum weight independent set of rectangles problem.
In addition, we design a framework which reduces the problem to solving instances 
where essentially the edge capacities are within a constant factor of each other. 
The techniques can be applied and combined in various ways. For instance, for practical purposes, we also show how our results can be used to obtain a constant factor approximation algorithm with a reasonable running time of only~$O(n^{4}\log n)$. We now go into more detail about these results, the new techniques we introduce, and give an outline of the paper.

Similar to many previous papers, for our main algorithm we partition the tasks into `small' and `large' tasks.
For the small tasks our main result is as follows:
For any $\gamma>0$ and $\epsilon>0$, we present a $(3+\epsilon)$-approximation algorithm for UFPP in Section~\ref{sec:small-tasks}, for the case where each task is $(1-\gamma)$-small.
We remark that a similar result was given by Chekuri et~al.~\cite{CEKApprox2009}, who gave an $O(\log(1/\gamma)/\gamma^3)$-approximation algorithm if each task is $(1-\gamma)$-small. Their result also applies to trees.
To prove our $(3+\epsilon)$-approximation, we introduce a novel framework in which the tasks are first 
grouped into smaller sets, according to their $b(i)$ values, 
such that the techniques for the NBA case can be applied.
So the resulting sets can be solved via relatively standard dynamic programming, LP-rounding, and network flow techniques. (This is similar to e.g.~\cite{CCKR2002,I-MCF-factor4-trees}.)
Solutions to these smaller sets leave a small amount 
of the capacity of each edge unused. In our framework we 
recombine these solutions into a feasible solution for all tasks.

Using the techniques developed for the $(3+\epsilon)$-approximation, in Section~\ref{sec:resaugmentation} we also give a result for \UFPP\ in the setting of \emph{resource augmentation}.
We give an algorithm that computes a $(2+\epsilon)$-approximative solution
which is feasible if we increase the capacity of each edge by a factor
of $1+\beta$, for arbitrarily small $\epsilon>0$ and $\beta>0$. Note that 
this algorithm works with arbitrary task sets and does not require the tasks to be small.

For our main approximation algorithm, it remains to handle the large tasks. For these, we present the following main result: for any integer $k\ge 2$, if all tasks are $\frac{1}{k}$-large, we give a $2k$-approximation algorithm in Section~\ref{sec:largetasks}. 
This is based on a geometric viewpoint of the problem:
we represent \UFPP\ instances by drawing a curve in the plane determined by the edge capacities, and representing tasks by axis-parallel rectangles, that are drawn as high as possible under this curve. The demand of a task determines the height of its rectangle, and the profit of a task determines the weight of the rectangle. 
Using a novel geometrically inspired dynamic program, we show that in polynomial time, 
a maximum weight set of pairwise non-intersecting rectangles can be found. 
Such a set corresponds to a feasible \UFPP\ solution. In addition, we show that when every task is $\frac{1}{k}$-large, this solution
yields a $2k$-approximative solution for \UFPP. 
With this dynamic program we contribute towards the well-studied problem of finding a Maximum Weight Independent Set of Rectangles (MWISR)~\cite{AKS1998,KMP1998,N2000}. Below we discuss this problem in more detail.

For our main result, we partition the tasks into $\frac{1}{2}$-small tasks and $\frac{1}{2}$-large tasks. 
For the first group, we apply the aforementioned $(3+\epsilon)$-approximation algorithm. 
For the second group, our second algorithm gives a 4-approximation. 
Returning the best solution of the two yields the $(7+\epsilon)$-approximation algorithm. The main algorithm is summarized in Section~\ref{sec:mainapproxalgos}. In addition, in Section~\ref{sec:mainapproxalgos} we show how our results can be combined to obtain a $O(n^4 \log n)$ time constant factor approximation algorithm for \UFPP, and we discuss how our results carry over to the generalization from a path to cycle networks, where we obtain a $(8+\epsilon)$-approximation algorithm.

Finally, we give an alternative proof of the strong NP-hardness of UFPP, which shows that a different restriction also remains strongly NP-hard. In the existing NP-hardness proofs~\cite{CWMX-ESA2010,DPS2010}, arbitrarily large demands are used in the reductions.
In Section~\ref{sec:hardnesssketch}, we prove that the problem is strongly NP-hard even for the restricted case where all demands are chosen from $\{1,2,3\}$ and capacities are uniform (RAP). 
Note that in contrast to our hardness result, it is known that in the slightly more restricted case where the capacities {\em and demands} are uniform, the problem admits a polynomial time algorithm: In that case, Arkin and Silverberg~\cite{AS1987} have shown that the problem can be
solved in time~$O(n^{2}\log n)$ by minimum-cost flow computations.

We end in Section~\ref{sec:discussion} with a discussion. In Section~\ref{sec:prelim}, notation and terminology are introduced. First, in the next subsection, we give more background on the many variants of UFPP that have been studied in the literature.

\subsection{Related Results}
\label{ssec:related}

As mentioned above, the restrictions of UFPP where demands are uniform (RAP) and where the No-Bottleneck Assumption $\max_i d_{i}\leq \min_e u_{e}$ holds (UFPP-NBA) have been well-studied, with the current best approximation algorithm for both problems being the $(2+\epsilon)$-approximation of Chekuri et~al.~\cite{I-MCF-factor4-trees}.
Both RAP and UFPP-NBA have been generalized in various ways: in a scheduling context, one may allow more freedom for choosing the start time $s_i$ and end time $t_i$ of a given task $i$.
Philips et~al.~\cite{PUW2000} obtain a~$6$-approximation algorithm
for such a generalization of RAP, by using LP-rounding techniques.
For a similar generalization, where for each task one out of a set of alternatives needs to be selected, 
Bar-Noy et~al.~\cite{BBFNS2000} provide a constant factor approximation algorithm using the local ratio technique. 
RAP and UFPP-NBA can be generalized in a different way, more related to network flows, by considering graphs other than a path.
In graphs other than trees, there may be different paths possible between a terminal pair $s_i$, $t_i$. However, a {\em single} path has to be chosen for each selected terminal pair. This generalization of UFPP to general graphs is called the {\em Unsplittable Flow Problem (UFP)}, or {\em UFP-NBA} if the NBA applies.
Baveja and Srinivasan~\cite{BS2000} provide an~$O(\sqrt{|E|})$-approximation algorithm for UFP-NBA (on all graphs), improving on various earlier results. A simpler combinatorial algorithm with the same guarantee was subsequently given by Azar and Regev~\cite{AR2006}.
Chakrabarti et~al.~\cite{CCGK2007} also give an approximation algorithm for all graphs.
In addition, they observed that an $\alpha$-approximation algorithm for UFPP-NBA gives an $(\alpha+1+\epsilon)$-approximation for UFP-NBA on cycles, for any $\epsilon>0$. This way they gave the first constant factor approximation algorithm for both UFPP-NBA and UFP-NBA on cycles.
Chekuri et~al.~\cite{I-MCF-factor4-trees} obtain a 48-approximation for UFP-NBA on trees. 

In addition, many hardness results are known for UFPP (not necessarily under the NBA) generalized to various graphs:
for general graphs, it is hard to approximate within a factor of~$\Omega(|E|^{1-\epsilon})$ unless $P=NP$~\cite{AR2006}, and for depth-3 trees the problem is APX-hard~\cite{GVY1997}. Hardness-of-approximation results are known even for the special case with unit demands and capacities (the \emph{Edge Disjoint Path Problem}); see~\cite{ACZ2005,AZ2005}.

When viewing UFPP as a packing problem, the corresponding {\em covering problem} has also been studied~\cite{CKRS2011ESA,CPRSK2011}. In that case, tasks have {\em costs} instead of profits, and the objective is to find a minimum cost set of tasks $F$, such that
for each edge, 
the sum of the demands of all tasks in $F$ that use this edge is {\em at least} its capacity.
Recently, Chakaravarthy~et~al.~\cite{CKRS2011ESA} designed a primal-dual 4-approximation algorithm
for this problem.

Recall that we reduced UFPP for large tasks to a special case of the
{Maximum (Weight) Independent Set of Rectangles}~(M(W)ISR) problem. 
In this problem, a collection of~$n$ axis-parallel rectangles is given and
the task is to find a maximum (weight) subset of disjoint rectangles.
For the unweighted case of this problem, a randomized~$O(\log\log n)$-approximation is known~\cite{CC2009}. 
For the weighted case, there are several $O(\log n)$-approximation algorithms~\cite{AKS1998,KMP1998,N2000}.
The algorithm by Erlebach~et~al.~\cite{EJS2001} gives a PTAS for the case that the ratio between height and width of rectangles is bounded
(note that this does not apply in the special case that we need here for approximating \UFPP). 
Our new dynamic programming technique might be useful for further research on this problem.

We remark that this approach for the large tasks is closely related to another well-studied variant of RAP: In adjacent resource scheduling problems, one wants to schedule a job on several machines
in parallel which must be \emph{contiguous}, that is, adjacent to each other. 
In other words, this is a variant of MWISR where rectangles are allowed to move vertically, within a given range.
Duin and van Sluis~\cite{DS2006} prove the decision variant of
scheduling tasks on contiguous machines to be strongly NP-complete.
RAP on contiguous machines has been considered under the name \emph{storage
allocation problem} (SAP), in which tasks are axis-aligned rectangles
that are only allowed to move vertically. 
Leonardi et~al.~\cite{LMV2000} provide a~$12$-approximation
algorithm for SAP and Bar-Yehuda~et~al.~\cite{BYBC2005} present a~$(2+\epsilon+1/(e-1))$-approximation
algorithm.

In this paper, we study \UFPP\ in the setting of
\emph{resource augmentation}. This means that we find a solution which
is feasible if we increase the capacity of each edge by a modest factor
of $(1+\epsilon)$. The paradigm of resource augmentation is very
popular in real-time scheduling. There, the augmented resource is
the speed of the machines. For instance, it is known that the natural
earliest deadline first policy (EDF) is guaranteed to work on $m$
machines with speed $2-1/m$ if the instance can be feasibly scheduled
on $m$ machines with unit speed~\cite{PSTW2008}. In addition, a matching
feasibility test is known~\cite{BMS2008}. For further examples of
resource augmentation results in real-time scheduling see \cite{ER2008,LT1999}.

\section{Preliminaries}\label{sec:prelim}

We assume that the vertices of the path $P=(V,E)$ are numbered $V=\{0,\ldots,m\}$, and $E=\{\{i,i+1\} \mid 0\le i\le m-1\}$.
We assume that the tasks are numbered $T=\{1,\ldots,n\}$.
Recall that tasks are characterized by two vertices~$s_i$ and~$t_i$ with $s_i<t_i$, and positive integer demand~$d_i$ and profit~$w_i$.

For each task~$i\in T$ we denote by $P_{i}\subseteq E$ the edge set of the subpath
of~$P$ from~$s_{i}$ to~$t_{i}$.
If $e\in P_i$, then task~$i$ is said to {\em use}~$e$.
For each edge $e$ we denote by $T_{e}\subseteq T$ the set of tasks
which use $e$. 
For a set of tasks $F$ we define its profit by $w(F):=\sum_{i\in F}w_{i}$.
Our objective is to find a set of tasks $F$ with
maximum profit such that $\sum_{i\in F\cap T_{e}}d_{i}\le u_{e}$ for each edge $e$. 
For each task
$i$ we define its \emph{bottleneck capacity }$b(i)$ by $b(i):=\min_{e\in P_{i}} u_{e}$.
An edge $e$ is called a \emph{bottleneck edge} for the task $i$
if $e\in P_{i}$ and $u_{e}=b(i)$. In addition, we define for every
task $i$ that $\ell(i):=b(i)-d_{i}$. 
The value $\ell(i)$ can be interpreted as the remaining capacity 
of a bottleneck edge of $i$ when $i$ is selected in a solution.
Consider a vertex
$v\in V$ and an edge $e\in E$ with $e=\{x,x+1\}$. We write $v<e$
(or $v>e$) if $v\le x$ (resp.\ $v\ge x+1$). For two edges $e=\{x,x+1\}$
and $e'=\{x',x'+1\}$ in $E$, we write $e<e'$ if $x<x'$ and $e\le e'$
if $x\le x'$. In other words, we interpret an edge $\{x,x+1\}$ simply
as a number between $x$ and $x+1$.
Without loss of generality, we will assume throughout this paper that
$u_e \ge 1$ for all edges $e$ and $d_i \ge 1$ for all tasks $i$; zero demands and capacities can easily be handled in a preprocessing step. 
Moreover,
observe that one can easily adjust any given instance to an equivalent 
instance in which each vertex is either a start or an end vertex of a task. 
Such an adjustment can be implemented in linear time and it hence 
does not dominate the running times of the algorithms presented in this paper. 
Therefore, we will henceforth assume that $m<2n$. 
Throughout this paper, we will use the notations defined above to
refer to the \UFPP\ instance currently under consideration. In the few cases where we consider multiple instances, the notations will be clear from the context.

We define an $\alpha$-approximation algorithm for a maximization problem to be 
a polynomial time algorithm which computes a feasible solution 
for a given instance such that its objective value is at least $\frac{1}{\alpha}$ times the
optimal value. 
Throughout this paper, for a subset of the tasks $F\subseteq T$, 
$OPT(F)$ denotes an optimal solution for the \UFPP\ instance restricted to the task set $F$.
The following simple fact shows how we can combine our approximation algorithms for different task subsets into one algorithm for all tasks. 

\smallskip
\begin{fact}
\label{fact:additive_approximation}
Consider a UFPP instance with task set $T$, and a partition $\{T_1,T_2\}$ of $T$. If for $i=1,2$, there exists an $\alpha_i$-approximation algorithm for the instance restricted to the tasks in $T_i$, then there exists an $(\alpha_1+\alpha_2)$-approximation algorithm for the entire instance.
\end{fact}
\begin{proof}
For $i=1,2$, let $ALG_i$ denote the solution returned by the approximation algorithm for the instance restricted to the tasks in $T_i$, and $OPT_i=OPT(T_i)$ an optimal solution.
Let $OPT=OPT(T)$ denote an optimal solution for all tasks.
So for $i=1,2$, $w(ALG(T_i))\ge \frac{1}{a_i}w(OPT(T_i))$.
The algorithm that returns the maximum profit solution of $ALG_1$ and $ALG_2$ has an approximation ratio of $\alpha_1+\alpha_2$, since
\begin{align*}
w(OPT) 
&= w(OPT\cap T_1)+w(OPT\cap T_2) 
 \leq w(OPT_1)+w(OPT_2) \\
& \leq \alpha_1 w(ALG_1) + \alpha_2 w(ALG_2)\\
& \leq \alpha_1\max\{w(ALG_1), w(ALG_2)\} + \alpha_2 \max\{w(ALG_1),w(ALG_2)\}\\
&\leq (\alpha_1+\alpha_2)\max\{w(ALG_1), w(ALG_2)\}. 
\end{align*}
\end{proof}

\section{\label{sec:small-tasks}Small Tasks}

In this section we present a $(3+\epsilon)$-approximation algorithm
for any set of tasks which are $(1-\gamma)$-small (for arbitrarily
small $\epsilon>0$ and $\gamma>0$). 
In our main algorithm (for a general
set of tasks) we will invoke this algorithm as a subroutine for all
tasks which are $\frac{1}{2}$-small.
Moreover, with a slight adjustment of introduced techniques we construct a polynomial time algorithm computing a $(2+\epsilon)$-approximative solution for the entire instance (not only small tasks) which is feasible if the capacities
of the edges are increased by a factor $1+\epsilon$ (resource augmentation), see Section~\ref{sec:resaugmentation}.

Our strategy is to define groups of tasks such that the 
bottleneck capacities of all tasks in one group are within a certain range.
This allows us to compute a feasible solution for each group, 
whose profit is at most a factor $3+\epsilon'$
smaller than the profit of an optimal solution for the group. In addition, each computed solution leaves a certain amount of capacity of every edge unused. We devise a framework
which combines solutions for a selection of groups into a feasible solution for the entire instance, in a way that yields a $(3+\epsilon)$-approximation 
(with an appropriate choice of $\epsilon'$ for the given $\epsilon$).

\subsection{Framework}
\label{ssec:framework}

We define the framework sketched above. We group the tasks into sets
according to their bottleneck capacities. Let
$\ell\in\N$ be a constant. We define $F^{k,\ell}:=\left\{ i\in T|2^{k}\le b(i)<2^{k+\ell}\right\} $
for each integer~$k$. Note
that this includes negative values for $k$, and that at most $\ell\cdot n$ sets are non-empty (only those will be relevant later). 
In the sequel, we will present
an algorithm which computes feasible solutions $ALG\left(F^{k,\ell}\right)\subseteq F^{k,\ell}$.
These solutions will satisfy the following properties.

\smallskip
\begin{definition} 
Consider a set $F^{k,\ell}$ and let $\alpha>0$ and $\beta \ge0$. A set $F\subseteq F^{k,\ell}$ is
called $(\alpha,\beta)$\emph{-approximative} if 
\begin{itemize}
\item $w(F)\ge\frac{1}{\alpha}\cdot w(OPT(F^{k,\ell}))$, and
\item $\sum_{i\in F\cap T_{e}}d_{i}\le u_{e}-\beta\cdot2^{k}$ for each
edge $e$ such that $T_{e}\cap F^{k,\ell}\ne\emptyset$.  (In particular, it is a feasible solution.)
\end{itemize}
An algorithm which computes \emph{$(\alpha,\beta)$-approximative}
sets in polynomial time is called an $(\alpha,\beta)$\emph{-approximation
algorithm}. We call the second condition the~\emph{modified capacity
constraint}. 
\end{definition}

\smallskip
Our framework consists of a procedure that turns an $(\alpha,\beta)$-approximation
algorithm for each set $F^{k,\ell}$ into a $\left(\frac{\ell+q}{\ell}\cdot\alpha\right)$-approximation
algorithm for all given tasks, where $q$ and $\beta$ are chosen such that~$\beta =2^{1-q}>0$.
Later, for our resource augmentation result (see Section~\ref{sec:resaugmentation}) we will work with 
$(\alpha,0)$-approximative sets and therefore,
some of the claims below will be proven more generally, allowing $\beta$ to be zero.

\smallskip
\begin{lemma}[Framework]
\label{lem:framework} Let $\ell \in \N$ and $q \in \N$ be constants and let $\beta :=2^{1-q}$.
Let the sets
$F^{k,\ell}$ be defined as stated above for an instance of \UFPP.
Assume we are given an \emph{$(\alpha,\beta)$}-approximation algorithm
for each set $F^{k,\ell}$ with running time $O(p(n))$ for a polynomial
$p$. Then there is a $\left(\frac{\ell+q}{\ell}\cdot\alpha\right)$-approximation
algorithm with running time $O(m\cdot p(n))$ for the set of \emph{all}
tasks.
\end{lemma}

\smallskip
Now we describe the algorithm that yields Lemma~\ref{lem:framework}.
Assume that we are given an $(\alpha,\beta)$-approximation algorithm which
computes solutions $ALG\left(F^{k,\ell}\right)\subseteq F^{k,\ell}$. The
key idea is that due to the unused edge-capacities of the sets $ALG\left(F^{k,\ell}\right)$,
the union of several of these sets still yields a feasible solution.
With an averaging argument we will show further
that the indices $k$ for the sets $ALG\left(F^{k,\ell}\right)$ that
we want to combine can be chosen such that the resulting set
is an $\left(\frac{\ell+q}{\ell}\cdot\alpha\right)$-approximation.
Formally, for each {\em offset} $c\in\{0,...,\ell+q-1\}$ we define an index set $\eta(c)=\{c+i\cdot(\ell+q)\mid i\in\Z\}$ (the values
$\ell$ and $q$ that it depends on will always be clear from the context).
For each $c\in\{0,...,\ell+q-1\}$ we compute the set $ALG(c)=\bigcup_{k\in\eta(c)}ALG(\F^{k,\ell})$.
We output the set~$ALG(c^{*})$ with maximum profit
among all sets $ALG(c)$. In Lemma~\ref{lem:weight-ALG(c*)} we will
prove that the resulting set is an $(\frac{\ell+q}{\ell}\cdot\alpha)$-approximation.
First, in Lemma~\ref{lem:ALG(c)-feasible} we will prove that each set $ALG(c)$
is feasible (using that $\beta \ge 2^{1-q}$). This requires the following property.

\smallskip
\begin{proposition}
\label{propo:demand-bottleneck-freecapacity}
Let $F$ be a feasible
UFPP solution 
such that $b(i)<u$ for all~$i\in F$.
Then for every edge $e$, $\sum_{i\in F\cap T_{e}}d_{i}<2u$.
\end{proposition}
\begin{proof} 
Let $e$ be an edge. If $u_{e}< 2u$, then the claim follows immediately.
Now suppose that $u_{e} \ge 2u$. Any task $i\in F\cap T_{e}$ must use
an edge whose capacity is less than $u$. In particular, it must use
either the closest edge $e_L$ to the left of $e$ or the closest edge $e_R$ to the right of $e$
whose capacity is less than $u$. 
The total demand of tasks in $F$ using $e_L$ is less than $u$, and the same holds for $e_R$.
It follows that the total capacity
used by tasks in $F\cap T_e$ must be less than $2u$.
\qquad\end{proof}

\smallskip
\begin{lemma}
\label{lem:ALG(c)-feasible} Let $\ell \in \N$ and $q \in \N$ be constants.
For each set $F^{k,\ell}$ let $ALG\left(F^{k,\ell}\right)$ be a
$(\alpha,\beta)$-approximative set with $\beta \ge 2^{1-q}>0$. Then 
for each $c\in\{0,...,\ell+q-1\}$ the
set $ALG(c)=\bigcup_{k\in\eta(c)}ALG(\F^{k,\ell})$ is feasible. \end{lemma}
 \begin{proof}
Consider a set $ALG(F^{k,\ell})$.
By definition of the $(\alpha,\beta)$-approximative sets, 
$ALG(F^{k,\ell})$ leaves $\beta \cdot 2^k \ge 2^{k+1-q}$ units of 
the capacity of every used edge free. 
Observe that this is at least twice the maximum bottleneck capacity of tasks in $F^{k-(\ell+q),\ell}$.
Therefore, by Proposition~\ref{propo:demand-bottleneck-freecapacity}, the set $ALG(F^{k,\ell}) \cup ALG(F^{k-(\ell+q),\ell})$ 
is feasible. In fact, it again leaves a fraction of the capacities free, which makes it possible to continue this argument for further sets $ALG(F^{k-i(\ell+q)})$ with $i\ge 2$, and prove that their union is feasible.

Formally, let $c\in\{0,...,\ell+q-1\}$ and let $e$ be an edge. Denote $P :=\ell+q$.
Let $\bar{k}$ be the largest integer in $\eta(c)$ such that $2^{\bar{k}}\le u_{e}$.
For every $x\in\eta(c)$, denote 
\[
U_{e}^{x}=\sum_{j\in T_{e}\cap ALG(F^{x,\ell})}d_{j}.
\]
Since $ALG(F^{\bar{k},\ell})$ is an \emph{$(\alpha,\beta)$}-approximation
and $2^{1-q}\le\beta$, we have that
\[
U_{e}^{\bar{k}}\le u_{e}-\beta\cdot2^{\bar{k}}\le u_{e}-2^{\bar{k}+1-q}.
\]
For every $i\ge1$, applying Proposition~\ref{propo:demand-bottleneck-freecapacity} for the modified capacities $u'_e=u_e-\beta\cdot 2^{\bar{k}-i\cdot P}$
shows that 
\[
U_{e}^{\bar{k}-i\cdot P}\le2(2^{\bar{k}-i\cdot P+\ell}-\beta\cdot2^{\bar{k}-i\cdot P})\le2^{\bar{k}+1-i\cdot P+\ell}-2^{\bar{k}+2-i\cdot P-q}.
\]
 Summarizing, we have that \begin{eqnarray*}
\sum_{j\in T_{e}\cap ALG(c)}d_{j} & = & \sum_{i=0}^{\infty}U_{e}^{\bar{k}-i\cdot P}\\
 & \le & u_{e}-2^{\bar{k}+1-q}+\sum_{i=1}^{\infty}(2^{\bar{k}+1-i\cdot P+\ell}-2^{\bar{k}+2-i\cdot P-q})\\
 & < & u_{e}+\sum_{i=1}^{\infty}2^{\bar{k}+1-(i-1)\cdot P-q}-\sum_{i=0}^{\infty}2^{\bar{k}+1-i\cdot P-q}=u_{e}.
 \end{eqnarray*}
 \end{proof}

\smallskip
\begin{lemma}
\label{lem:weight-ALG(c*)}
Let $\ell \in \N$ and $q \in \N$ be constants.
For each set $F^{k,\ell}$ let $ALG\left(F^{k,\ell}\right)$ be a
$(\alpha,\beta)$-approximative set with $\beta \ge 0$.
Then for the offset $c^{*}$ which maximizes $w(ALG(c))=\sum_{k\in\eta(c)}w(ALG(\F^{k,\ell}))$ it holds that $w(ALG(c^{*}))\ge\frac{\ell}{\ell+q}\cdot\frac{1}{\alpha}\cdot w(OPT)$,
where $OPT$ denotes an optimal solution of the given instance of UFPP.
\end{lemma}
\begin{proof}
Every task is included in $\ell$ different sets $F^{k,\ell}$.
Using this fact, we calculate that 
\[
\sum_{c=0}^{\ell+q-1}w(ALG(c)) \ge
\sum_{c=0}^{\ell+q-1}\sum_{k\in\eta(c)}\frac{1}{\alpha}\cdot w\left(OPT\left(\F^{k,\ell}\right)\right) =
\]
\[
\frac{1}{\alpha}\cdot\sum_{k\in\Z}w\left(OPT\left(\F^{k,\ell}\right)\right)
 \ge  
\frac{1}{\alpha}\cdot\sum_{k\in\Z}w\left(OPT\cap\F^{k,\ell}\right)
= 
\frac{\ell}{\alpha}\cdot w(OPT).
\]
So there must be a value $c$ such that $w(ALG(c))\ge\frac{\ell}{\ell+q}\cdot\frac{1}{\alpha}\cdot w(OPT)$.
In particular, this holds for $c^*$.
\qquad\end{proof}

\begin{pfof}{Lemma~\ref{lem:framework}}
 Lemma~\ref{lem:ALG(c)-feasible} shows that $ALG(c^{*})$ is
feasible and Lemma~\ref{lem:weight-ALG(c*)} shows that $ALG(c^{*})$
is a $\left(\frac{\ell+q}{\ell}\cdot\alpha\right)$-approximation.
For computing $ALG(c^{*})$ we need to compute the set $ALG\left(F^{k,\ell}\right)$
for each relevant value~$k$. There are at most $m\ell\in O(m)$
relevant values $k$. Finding the optimal offset $c^{*}$ can be done
in $O(m)$ steps. This yields an overall running time of $O\left(m\cdot p(n)\right)$
(recall that $p(n)$ is the polynomial bounding the running time needed to compute the sets
$ALG\left(F^{k,\ell}\right)$). 
\end{pfof}

\subsection{An Approximation Algorithm for Small Tasks}
\label{ssec:ab-approx}

Now that we have developed the framework to translate $(\alpha,\beta)$-approximation algorithms for the sets $F^{k,\ell}$ into an approximation algorithm for the entire instance (Lemma~\ref{lem:framework}), it remains to present such an $(\alpha,\beta)$-approximation algorithm. In this section, we present a $(2+\frac{1+\epsilon}{1-\beta},\beta)$-approximation algorithm for sets 
$F^{k,\ell}$ in which all tasks are $(1-2\beta)$-small (for arbitrarily small $\beta \ge0$). Together with our framework of Lemma~\ref{lem:framework} this yields a $(3+\epsilon)$-approximation algorithm for UFPP for the case that all tasks are $(1-\gamma)$-small, for arbitrary $\epsilon>0$ and $\gamma>0$. To get some intuition, the reader can think of~$\beta$ being equal to $\min \{ \epsilon, \gamma/2 \}$.

Suppose we are given a set $F^{k,\ell}$ with only $(1-2\beta)$-small tasks. In order to derive the mentioned $(2+\frac{1+\epsilon}{1-\beta},\beta)$-approximation algorithm, we choose a value $\delta>0$ and split the set into 
$\delta$-small tasks (\emph{tiny} tasks) and tasks which are $\delta$-large but $(1-2\beta)$-small (\emph{medium} tasks). We define $\delta$ such that for the tiny tasks there is a  $(\frac{1+\epsilon}{1-\beta},\beta)$-approximation algorithm, presented in the following subsection. For the medium tasks, we give a $(2,\beta)$-approximation algorithm in Section~\ref{sec:medium-tasks}.

\subsubsection{An Approximation Algorithm for Tiny Tasks\label{sec:tiny-tasks}}
We show that for given $\epsilon>0$ and $\beta \ge 0$, there is a $\delta>0$ such that if all tasks are $\delta$-small, then for each set $F^{k,\ell}$ there is a $(\frac{1+\epsilon}{1-\beta},\beta)$-approximation algorithm. The key idea is to use linear programming techniques and a result by 
Chekuri et al.~\cite{I-MCF-factor4-trees} about the integrality gap of the canonical LP-relaxation of UFPP under the no-bottleneck assumption (NBA). UFPP with a task set~$T$ can be formulated in a straightforward way as an integer linear program:
\begin{align*}
IP:\qquad\qquad \max\quad \sum_{i\in T}w_{i}\cdot x_{i}\\
\mbox{s.t.} \sum_{i\in  T_{e}}x_{i}\cdot d_{i} & \le u_{e} &  & \forall\: e\in E\\
 x_{i} & \in\{0,1\} &  & \forall\: i\in T
\end{align*}
The LP relaxation is obtained by replacing the constraint $x_i\in \{0,1\}$ by $0\le x_i\le 1$. Chekuri et al.~\cite{I-MCF-factor4-trees} gave an algorithm for UFPP instances which satisfy the NBA (i.e. $\max_i d_{i}\leq \min_e u_{e}$), in which all tasks are $\delta$-small, which returns a UFPP solution that is at most a factor $f(\delta)$ worse than the optimum of the LP relaxation. Here $f(\delta)$ is a function for which the limit is 1 as $\delta$ approaches zero. In other words: The integrality gap of the canonical LP-relaxation is $1+\epsilon$ if the NBA holds and all tasks are sufficiently small.

This result can be used for sets $F^{k,\ell}$: By definition, tasks in $F^{k,\ell}$ use only edges~$e$ with $u_e\ge 2^k$. Call these {\em relevant edges}.
It is therefore possible to choose the value~$\delta$ small enough to ensure that the NBA holds, when considering only the relevant edges and $\delta$-small tasks in $F^{k,\ell}$.
Furthermore, modifying the capacities by choosing $u'_e=u_e-\beta\cdot 2^k$ decreases the capacities of relevant edges at most by a factor $1-\beta$. Therefore, the optimal value of the LP relaxation also becomes at most a factor $1-\beta$ smaller.
These are the key ideas to prove the following lemma:

\smallskip
\begin{lemma}
\label{lem:F^k-LP}
For every combination of constants $\epsilon>0$, $0\le\beta<1$, and $\ell\in\N$, there exists a $\delta>0$ such that if all tasks are $\delta$-small, then for each set $F^{k,\ell}$ there is a $(\frac{1+\epsilon}{1-\beta},\beta)$-approximation algorithm.
\end{lemma}

\smallskip
In the remainder of this subsection we prove Lemma~\ref{lem:F^k-LP}.
Assume we are given constants $\epsilon>0$, $\ell\in\N$, and $\beta$ with $0\le \beta <1$.
For an instance $\III$ of UFPP, we denote by~$LP(\III)$ the natural 
LP-relaxation of the IP-formulation given above, where each constraint $x_{i} \in\{0,1\}$ is replaced by $0\le x_i\le 1$.
By $\lpopt(\III)$ we denote the optimum value of the LP. 
We define~$f(\delta') = \frac{1+\sqrt{\delta'}}{1-\sqrt{\delta'}-\delta'}$.
The following result is proved by Chekuri et al.~\cite{I-MCF-factor4-trees}, although an exact analysis of the running time is not given. We observe that their algorithm admits a $O(n^3\log n)$ implementation.

\smallskip
\begin{lemma}
\label{lem:chekuri}
Consider a UFPP instance $\III$ for which the NBA holds, in which all tasks are $\delta'$-small, with $\delta'\le \frac{3-\sqrt{5}}{2}$.
Then in time $O(n^3 \log n)$, a feasible UFPP solution $ALG$ for $\III$ can be computed with $w(ALG)\ge f(\delta')^{-1}\cdot \lpopt(\III)$.
\end{lemma}
\begin{proof}
The algorithm of Chekuri et~al.\ works as follows. 
The tasks are partitioned into at most~$n$ groups, depending on their demands. 
The demands and the capacity are scaled such that a problem with uniform demands and uniform capacities is obtained. 
Since the demands and capacities are uniform, this can be solved optimally in time~$O(n^2\log n)$, using the algorithm 
by Arkin and Silverberg~\cite{AS1987}.
Then Chekuri et al. show that combining the solutions of each group 
yields a feasible solution.
Furthermore, they have shown in~\cite[Corollary~3.4]{I-MCF-factor4-trees}
that the obtained solution is at most a factor $f(\delta')$ worse than the optimal LP solution.
\qquad\end{proof}

\smallskip
\begin{lemma}
\label{lem:LPtechnicalversion}
Let $\delta>0$, $\delta'>0$, $\beta \ge 0$ and $\ell\in \N$ be constants such that $\delta\le \frac{1-\beta}{2^\ell}$ and $\delta'=\frac{\delta}{1-\beta}\le \frac{3-\sqrt{5}}{2}$.
If all tasks in $F^{k,\ell}$ are $\delta$-small, then in time $O(n^3\log n)$, a solution $ALG$ for $F^{k,\ell}$ can be computed that is $(\frac{f(\delta')}{1-\beta},\beta)$-approximative.
\end{lemma}
\begin{proof}
Consider the UFPP instance $\III$ that is obtained by restricting the instance to the tasks in $F^{k,\ell}$, and only considering edges used by these tasks.
So for every task $i$, it holds that $2^k\le b(i)\le 2^{k+\ell}$, and $d_i\le \delta\cdot b(i)$. For every edge $e$, we have $u_e\ge 2^k$.

Now construct the instance $\III'$ from $\III$ by modifying the capacities as follows: $u'_e:=u_e-\beta\cdot 2^k$, for each edge $e$. The instance $\III'$ contains the same task set as $\III$, without modifications.
For a task $i$, by $b(i)$ and $b'(i)$ we denote its bottleneck capacity in $\III$ and $\III'$, respectively. Hence $b'(i)=b(i)-\beta\cdot 2^k$.

We will first argue that the algorithm from Lemma~\ref{lem:chekuri} may be applied to $\III'$.
First, we note that since in each task is $\delta$-small with respect to the original capacities, 
under the modified capacities each task is still $\delta'$-small:
\begin{eqnarray*}
d_i \leq \delta \cdot b(i) = \frac{\delta}{(1-\beta)} \cdot (1-\beta) \cdot b(i) 
\leq \delta' \cdot (b(i) - \beta \cdot 2^k)=\delta' \cdot b'(i).
\end{eqnarray*}
Recall that we required that $\delta'\le \frac{3-\sqrt{5}}{2}$, so this condition of Lemma~\ref{lem:chekuri} is satisfied.
Now we argue that for the modified capacities, the NBA holds.
Since~$\delta \leq (1-\beta)/2^\ell$, 
for all tasks~$i$ and edges $e$ it holds that 
$d_i\leq \delta\cdot 2^{k+\ell} \leq 2^k - \beta \cdot 2^k \leq u_e - \beta \cdot 2^k=u'_e$.
This shows that the NBA also holds with respect to the modified capacities.

Hence we may apply the algorithm from Lemma~\ref{lem:chekuri} to obtain a feasible solution $ALG$ for $\III'$, with
$w(ALG)\ge f(\delta')^{-1}\lpopt(\III')$.
We argue that~$ \lpopt(\III') \ge(1-\beta)\cdot \lpopt(\III)$:
Since for every edge $e$ it holds that $u_e\ge 2^k$, it follows that
$u'_e=u_e-\beta\cdot 2^k\ge (1-\beta) u_e$.
Thus, if we take a feasible solution to $LP(\III)$, and scale all the variables $x_i$ by a factor $(1-\beta)$, we obtain a feasible solution to $LP(\III')$, in which the objective value has also been scaled by a factor $(1-\beta)$. 
This gives that
\begin{eqnarray*}
w(ALG)  
& \stackrel{\textup{Lemma \ref{lem:chekuri}}}{\ge} 
& f(\delta')^{-1}\cdot  \lpopt(\III')\\
& \ge & f(\delta')^{-1}\cdot(1-\beta)\cdot \lpopt(\III)\\
& \ge & \left( f(\delta')\cdot \frac{1}{1-\beta}\right)^{-1}\cdot OPT(\III),
\end{eqnarray*}
where $OPT(\III)$ is an optimal (integer) UFPP solution for the instance $\III$.
\qquad\end{proof}

The proof of Lemma~\ref{lem:F^k-LP} now easily follows:

\begin{pfof}{Lemma~\ref{lem:F^k-LP}}
Observe that the limit of $f(\delta')$ is 1 as $\delta'$ approaches zero.
Hence for every $\epsilon>0$, $0\le \beta<1$ and $\ell\in \N$, we can choose a $\delta'>0$ and $\delta:=(1-\beta)\delta'$ such that $f(\delta')\le 1+\epsilon$, 
$\delta\le \frac{1-\beta}{2^\ell}$, and $\delta'\le \frac{3-\sqrt{5}}{2}$.
Then Lemma~\ref{lem:LPtechnicalversion} shows that if all tasks are $\delta$-small, then in time $O(n^3\log n)$, a solution $ALG$ for $F^{k,\ell}$ can be computed that is $(\frac{1+\epsilon}{1-\beta},\beta)$-approximative.
\end{pfof}

\subsubsection{An Approximation Algorithm for Medium Tasks}\label{sec:medium-tasks}

It remains now to find a $(2,\beta)$-approximation algorithm for tasks in $F^{k,\ell}$ that are both $\delta$-large and $(1-2\beta)$-small (for the $\delta$ we obtained from Lemma~\ref{lem:F^k-LP}), for arbitrarily small $\beta \ge 0$.
When restricting to sets $F^{k,\ell}$ with only $\delta$-large tasks, the essential property is that for any edge $e$ and any feasible solution $F$, there are at most $O\!\left(\frac{2^{\ell}}{\delta}\right)$ tasks in $F$ that use~$e$.
This property allows for a straightforward dynamic program to be used to compute an optimal solution (see e.g.~\cite{CCKR2002,CCGK2007}). This can be turned into a $(2,\beta)$-approximate solution: since tasks are $(1-2\beta)$-small, it can be shown that in polynomial time, any solution $F$ can be partitioned into two sets which are both feasible for the modified capacities. Using these ideas, we will prove the following lemma in the remainder of this subsection.

\smallskip
\begin{lemma}
\label{lem:F^k-DP}
Let $\beta\ge 0$, $\delta>0$ and $\ell\in\N$ be constants and assume we are given an instance of \UFPP\ in which
all tasks are both $(1-2\beta)$-small and $\delta$-large. 
Then, for each set $F^{k,\ell}$
there is a $(2,\beta)$-approximation algorithm.
\end{lemma}

\smallskip
Suppose we are given a set $F^{k,\ell}$ whose tasks are all $\delta$-large and $(1-2\beta)$-small. 

\smallskip
\begin{proposition}
\label{propo:DPbound}
Let $F\subseteq F^{k,\ell}$ be a feasible solution in which all tasks are $\delta$-large. Then for any edge $e$, at most $2^{\ell+1}/\delta$ tasks in $F$ use $e$.
\end{proposition}
\begin{proof}
Let $F$ be a feasible solution.
For each task $i\in\F^{k,\ell}$ it holds that $b(i)\ge2^{k}$.
In addition, all tasks are $\delta$-large. Therefore, for all tasks $i\in F$ it holds that $d_i\ge \delta \cdot 2^k$.
For every task $i\in F$, $b(i)<2^{k+\ell}$. So according to Proposition~\ref{propo:demand-bottleneck-freecapacity}, for every edge $e$ it holds that
$\sum_{i\in F\cap T_{e}}d_{i}<2^{k+\ell+1}$. 
Therefore, at most $\frac{2^{k+\ell+1}}{\delta \cdot 2^k}=\frac{2^{\ell+1}}{\delta}$ tasks in $F$ use $e$.
\end{proof}

Due to Proposition~\ref{propo:DPbound}, there is a straight forward dynamic programming algorithm for computing an optimal solution~\cite{CCKR2002,CCGK2007}.
(We remark that the running time is $n^{O\left(2^{\ell}/\delta\right)}$.)

\smallskip
\begin{proposition}
\label{propo:DPpoly}
For constant $\delta>0$ and $\ell\in \N$, if all tasks are $\delta$-large, an optimal solution for $F^{k,\ell}$ can be computed in polynomial time.
\end{proposition}

\begin{pfsk}
By Proposition~\ref{propo:DPbound}, each edge can be used by at most $2^\ell / \delta$ tasks. Hence, for each edge $e$ there can be at most $n^{O\left(2^{\ell}/\delta\right)}$ combinations of tasks which use $e$ in an optimal solution.
We enumerate all these combinations for each edge $e$. For each of these, we establish a dynamic programming cell which stores the maximum profit one can obtain from tasks $i$ with $s_i<e$, given the respective combination of tasks that use $e$. The correct values for these cells can be computed by iterating through the edges of the path from left to right. 
See~\cite{CCKR2002,CCGK2007} for details.
\end{pfsk}

\smallskip
\begin{lemma}
\label{lem:DP_two_partition}
Let $\beta \ge 0$ and let $F$ be a feasible solution for $F^{k,\ell}$ in which all tasks are $(1-2\beta)$-small. Then in time $O(n^2)$, $F$ can be partitioned into two sets $H^1$ and $H^2$ which are both feasible for the modified capacity constraints $u'_e:=u_e-\beta\cdot 2^k$.
\end{lemma}
\begin{proof}
Note that the claim is trivially true if $\beta=0$ by setting $H^1:=F$ and $H^2:=\emptyset$, so now assume $\beta>0$.
We initialize two sets $H^{1}:=\emptyset$, $H^{2}:=\emptyset$. Assume
that the tasks in $F$ 
are ordered
such that the start vertices are non-decreasing. 
We consider the tasks in $F$  
in this
order. In the $i$-th iteration we take the task~$i$. We add $i$
to a set $H^{\ell}$ with $\ell\in\{1,2\}$ such that $H^{\ell}\cup\{i\}$
obeys the modified capacity constraint, i.e., it leaves a free capacity
of $\beta\cdot2^{k}$ in each edge.

It remains to show that indeed either $H^{1}\cup\{i\}$ or $H^{2}\cup\{i\}$
obeys the modified capacity constraint. Assume to the contrary that
neither $H^{1}\cup\{i\}$ nor $H^{2}\cup\{i\}$ obey the modified
capacity constraint. Then there are edges $e,e'$ on the path of $i$
such that 
\begin{equation}
\sum_{i'\in(H^{1}\cup\{i\})\cap T_{e}}d_{i'}>u_{e}-\beta\cdot2^{k}\label{ineq:demand-violation1}
\end{equation}
 and
\begin{equation}
\sum_{i'\in(H^{2}\cup\{i\})\cap T_{e'}}d_{i'}>u_{e'}-\beta\cdot2^{k}.\label{ineq:demand-violation2}
\end{equation}
 Inequality~\eqref{ineq:demand-violation1} implies that $\sum_{i'\in H^{1}\cap T_{e}}d_{i'}>u_{e}-\beta\cdot2^{k}-d_{i}$.
Inequality~\eqref{ineq:demand-violation2} gives that $\sum_{i'\in H^{2}\cap T_{e'}}d_{i'}>u_{e'}-\beta\cdot2^{k}-d_{i}$.
Assume w.l.o.g.~that $e<e'$ or $e=e'$. Recall that we considered
the tasks by non-decreasing start index. Hence, all tasks in $(H^{2}\cup\{i\})\cap T_{e'}$
use $e$ as well. For the next calculation we need that
\[
d_{i}\le(1-2\beta)b(i)\le u_{e'}(1-2\beta)
\]
and hence $u_{e'}-d_{i}\ge2\beta\cdot u_{e'}$. Also, note that $u_{e'}\ge2^{k}$
since $i\in F\subseteq\F^{k,\ell}$. We
calculate that

\begin{eqnarray*}
u_{e} & \ge & \left(\sum_{i'\in H^{1}\cap T_{e}}d_{i'}\right)+\left(\sum_{i'\in H^{2}\cap T_{e'}}d_{i'}\right)+d_{i}\\
 & > & \left(u_{e}-\beta\cdot2^{k}-d_{i}\right)+\left(u_{e'}-\beta\cdot2^{k}-d_{i}\right)+d_{i}\\
 & = & u_{e}+u_{e'}-d_{i}-2\beta\cdot2^{k}\\
 & \ge & u_{e}+2\beta u_{e'}-2\beta\cdot2^{k}\\
 & \ge & u_{e}.\end{eqnarray*}
 This is a contradiction. Hence, task $i$ can be added to one of
the sets $H^{\ell}$ such that $H^{\ell}\cup\{i\}$ still obeys the
capacity constraint. 
When computing $H^1$ and $H^2$ 
we need to check
for each task in $F$ 
whether adding
it to one of the sets violates the modified capacity constraint in
one of the edges. Since w.l.o.g. $m<2n$, this check can be done in
$O(n)$ time for each task. There are $n$ tasks in total, and hence
the entire procedure can be implemented in $O(n^{2})$. 
\qquad
\end{proof}

Now we can prove Lemma~\ref{lem:F^k-DP}.

\begin{pfof}{Lemma~\ref{lem:F^k-DP}}
Since  the tasks are $\delta$-large, we can compute an optimal solution $OPT(F^{k,\ell})$ in polynomial time (Proposition~\ref{propo:DPpoly}).
Since tasks are $(1-2\beta)$-small, in time $O(n^2)$, the solution $OPT(F^{k,\ell})$ can be partitioned into two solutions $H^1$ and $H^2$ that obey the modified capacity constraint (Lemma~\ref{lem:DP_two_partition}). Returning the solution of these two with maximum profit then yields a $(2,\beta)$-approximation for $F^{k,\ell}$.
\end{pfof}

\subsubsection{An Approximation Algorithm for $(1-\gamma)$-Small Tasks}

By combining Lemmas~\ref{lem:F^k-LP} and \ref{lem:F^k-DP}, we obtain an $(\alpha,\beta)$-approximation algorithm for small tasks.

\smallskip
\begin{lemma}
\label{lem:small-tasks-framework-lemma}
Let $\epsilon>0$, $\beta \ge0$, and $\ell\in\N$ be constants with $\beta < 1$ and assume we are given an instance of \UFPP\ in which
all tasks are $(1-2\beta)$-small. Then, for each set $F^{k,\ell}$
there is a $(2+\frac{1+\epsilon}{1-\beta},\beta)$-approximation algorithm.
\end{lemma}
\begin{proof}
Given $\epsilon$, $\beta$ and $\ell$, Lemma~\ref{lem:F^k-LP} shows that there exists a $\delta>0$ such that for the $\delta$-small tasks we have an $(\frac{1+\epsilon}{1-\beta},\beta)$-approximation algorithm (for each set $F^{k,\ell}$). The remaining tasks are both $\delta$-large and $(1-2\beta)$-small. 
If $\delta\ge 1-2\beta$, we are done. Otherwise, Lemma~\ref{lem:F^k-DP} shows that we have an $(2,\beta)$-approximation algorithm for these (for each set $F^{k,\ell}$). Together this gives an $(2+\frac{1+\epsilon}{1-\beta},\beta)$-approximation algorithm for each $F^{k,\ell}$ (observe that Fact~\ref{fact:additive_approximation} also applies to $(\alpha,\beta)$-approximation algorithms).
\qquad\end{proof}

Using the above lemma and our framework (Lemma~\ref{lem:framework}), we obtain our main result for small tasks.

\smallskip
\begin{theorem}
\label{thm:small-3+eps}
For any pair of constants $\epsilon>0$ and $\gamma>0$, 
there is a polynomial time $(3+\epsilon)$-approximation algorithm for \UFPP\ instances in which all tasks are $(1-\gamma)$-small.
\end{theorem}
\begin{proof}
Choose $\ell\in \N$, $q\in \N$, $\beta:=2^{1-q}$ and $\epsilon'>0$ such that 
$1-2\beta\ge 1-\gamma$ (hence all tasks are $(1-2\beta)$-small), and 
$
\frac{\ell+q}{\ell} \cdot \left(2+\frac{1+\epsilon'}{1-\beta}\right)\le 3+\epsilon$,
which is always possible
\footnote{
For instance, choose $\beta$ small enough such that $1/(1-\beta)\le 1+\epsilon'$, $1-2\beta \ge 1-\gamma$, 
and that there is an integer $q\ge 2$ with $\beta=2^{1-q}$. Now choose $\ell \in \N$ such that
$\frac{\ell +q}{\ell} \le 1+\epsilon'$. We obtain an approximation factor of 
$(2+\frac{1+\epsilon'}{1-\beta})\cdot \frac{\ell +q}{\ell} \le (2+(1+\epsilon')^2)\cdot (1+\epsilon')$
which is at most $3+\epsilon$ if $\epsilon'$ is sufficiently small.
}. 
Now, combining the framework using the chosen $\ell$ and $q$ (Lemma~\ref{lem:framework}) with $(2+\frac{1+\epsilon'}{1-\beta},\beta)$-approximation algorithms for every $F^{k,\ell}$ (Lemma~\ref{lem:small-tasks-framework-lemma}) yields a $(3+\epsilon)$-approximation algorithm.
\qquad\end{proof}

\subsection{Resource Augmentation}
\label{sec:resaugmentation}

Using the techniques derived above, we now describe a polynomial time algorithm that computes a set of tasks $T'\subseteq T$ such that $w(T')\ge (2+\epsilon)^{-1}\cdot w(OPT)$, and $T'$ is feasible 
if the capacity of every edge is increased by a factor $1+\epsilon$. Note that for this result we do not require that the tasks $T$ are $(1-\gamma)$-small. 

The main idea is the following: from the results in Sections~\ref{sec:tiny-tasks} and \ref{sec:medium-tasks} we will conclude that there are $(1+\epsilon,0)$ and $(1,0)$-approximation algorithms for
the tiny tasks and all remaining tasks, respectively. Combining these, we obtain a $(2+\epsilon,0)$-approximation algorithm for each set $F^{k,\ell}$ (without any further conditions on its tasks) in Proposition~\ref{propo:2-eps-0-approx}.
We apply our framework with the sets computed by this algorithm. In Lemma~\ref{lem:RA-ALG(c)-feasible} we show that the union $ALG(c)=\bigcup_{k\in\eta(c)}ALG(\F^{k,\ell})$ is feasible when the capacities of the edges are increased by a factor $1+2^{2-q}$ (this lemma takes the role of Lemma~\ref{lem:ALG(c)-feasible} from the original framework).

The first step is to establish the approximation algorithms for the sets $ALG(F^{k,\ell})$. 

\smallskip
\begin{proposition}\label{propo:2-eps-0-approx}
Let $\epsilon>0$. There is a $(2+\epsilon,0)$-approximation algorithm for each set $F^{k,\ell}$ which runs in polynomial time.
\end{proposition}
\begin{proof}
Using Lemma~\ref{lem:F^k-LP} with $\beta=0$ yields that for all $\epsilon>0$ there is a $\delta>0$ such that there is a $(1+\epsilon,0)$-approximation algorithm for each set $F^{k,\ell}$ which consists only of $\delta$-small tasks. 
Proposition~\ref{propo:DPpoly} implies that for each fixed $\delta>0$ there is a polynomial time $(1,0)$-approximation algorithm (i.e., an optimal algorithm in the usual sense) for sets
$F^{k,\ell}$ which consist only of $\delta$-large tasks. Using Fact~\ref{fact:additive_approximation} this yields a $(2+\epsilon,0)$-approximation algorithm for each set $F^{k,\ell}$.
\qquad\end{proof}

The next lemma is an adjusted version of Lemma~\ref{lem:ALG(c)-feasible}. In contrast to the latter, here the increased capacity allows us to combine the computed sets $ALG(F^{k,\ell})$ to a globally feasible solution,
without requiring that every solution should leave a fraction of the capacity free. 
Recall that $\eta(c)=\{c+i\cdot(\ell+q)|i\in\mathbb{Z}\}$.

\smallskip
\begin{lemma}\label{lem:RA-ALG(c)-feasible}
Let $q\in\mathbb{N}$, $\ell\in\mathbb{N}$, and $c\in\{0,...,\ell+q-1\}$.
For each $k$, let $ALG^{k,\ell}\subseteq\F^{k,\ell}$ be a feasible
solution and define $\ensuremath{ALG(c)=\bigcup_{k\in\eta(c)}ALG(\F^{k,\ell})}$. 
Then for each edge $e$ it holds that 
\[
\sum_{i\in ALG(c)\cap T_{e}}d_{i}\le u_{e}\cdot(1+2^{2-q}).
\]
\end{lemma}
 \begin{proof}
Let $\bar{k}$ be the largest integer in $\eta(c)$ such that $2^{\bar{k}}\le u_{e}$.
Let $P:=\ell+q$. From Proposition~\ref{propo:demand-bottleneck-freecapacity} we conclude that $\sum_{i\in ALG^{k,\ell}\cap T_{e}}d_{i}\le2\cdot2^{k+\ell}$
for each $k<\bar{k}$. This implies that

\begin{eqnarray*}
\sum_{i\in ALG(c)\cap T_{e}}d_{i} & = & 
\sum_{i\in ALG^{\bar{k},\ell}\cap T_{e}}d_{i}+
\sum_{i=1}^{\infty}\sum_{j\in ALG^{\bar{k}-i\cdot P,\ell}\cap T_{e}}d_{j}\\
 & \le & u_{e}+\sum_{i=1}^{\infty}2\cdot2^{\bar{k}-i\cdot P+\ell}\\
 & = & u_{e}+2^{1+\bar{k}+\ell}\sum_{i=1}^{\infty}(2^{P})^{-i}\\
 & \le & u_{e}+2^{1+\ell}\cdot u_{e}\cdot\frac{1}{2^{P}-1}\\
 & =   & u_e\left(1+\frac{2^{1+\ell}}{2^{\ell+q}-1}\right)\\
 & \le & u_{e}\cdot(1+2^{2-q}). 
\end{eqnarray*}
 \end{proof}

The following lemma establishes our framework (see Lemma \ref{lem:framework}) in the setting of resource augmentation.

\smallskip
\begin{lemma}
\label{lem:framework-RA} Let $\III$ be a \UFPP\  instance. 
Let $\ell\in\mathbb{N}$ and $q\in\mathbb{N}$
be constants. If polynomial time $(\alpha,0)$-approximation
algorithms exist for each set $F^{k,\ell}$, then there is a polynomial time
algorithm that computes a set of tasks whose profit is at least $\alpha^{-1}\cdot\frac{\ell}{\ell+q}\cdot w(OPT(\III))$,
and which is feasible if the capacity of every edge $e$ is modified
to $u_{e}(1+2^{2-q})$.
\end{lemma}
\begin{proof}
The algorithm works as follows. For each $c\in\{0,...,\ell+q-1\}$
we define $\ensuremath{ALG(c):=\bigcup_{k\in\eta(c)}ALG(\F^{k,\ell})}$.
We output the set $ALG(c^{*})$ with maximum profit among all sets
$ALG(c)$. Lemma~\ref{lem:weight-ALG(c*)} implies that $w(ALG(c^{*}))\ge\frac{\ell}{\ell+q}\cdot \frac{1}{\alpha} \cdot w(OPT(\III))$.
Lemma~\ref{lem:RA-ALG(c)-feasible} yields that $ALG(c^{*})$ is
feasible if the capacity of each edge $e$ is changed to $u_{e}(1+2^{2-q})$. \qquad\end{proof}

The previous lemmas can be combined to obtain the main result of this section:

\smallskip
\begin{theorem}
Let $\epsilon>0$, $\beta>0$. Given an instance $\III$ of the UFPP
problem. There is a polynomial time algorithm that computes a set
of tasks whose profit is at least $(2+\epsilon)^{-1}\cdot w(OPT(\III))$,
and which is feasible if the capacity of every edge $e$ is modified
to $u_{e}(1+\beta)$.\end{theorem}
\begin{proof}
We choose $\ell\in\mathbb{N}$, $q\in\mathbb{N}$, and $\epsilon'>0$
such that $2^{2-q}\le\beta$ and $\frac{\ell+q}{\ell}\cdot(2+\epsilon')\le2+\epsilon$, 
which is always possible (similar as in the proof of Theorem~\ref{thm:small-3+eps}). The claim follows by combining the
framework of Lemma~\ref{lem:framework-RA} with the $(2+\epsilon',0)$-approximation
algorithms from Proposition~\ref{propo:2-eps-0-approx}. 
\qquad\end{proof}

\section{Large Tasks}
\label{sec:largetasks}

In this section we provide a polynomial time $2k$-approximation algorithm
for instances consisting of only~$1/k$-large tasks.
In our main algorithm (in Section~\ref{sec:mainapproxalgos}), this is used as a~4-approximation algorithm for the set of~$1/2$-large tasks.
The main idea is to restrict to UFPP solutions of a certain form:
we will represent tasks by rectangles drawn in the plane, and compute an
independent set of rectangles of maximum weight (profit).

By $(x_1,y_1,x_2,y_2)$ we will denote the axis-parallel rectangle in the 
plane $\mathbb{R}^2$ with upper left point $(x_1,y_1)$ and lower right point $(x_2,y_2)$.
We will call two rectangles $(x_1,y_1,x_2,y_2)$ and $(x'_1,y'_1,x'_2,y'_2)$ {\em compatible} if they do not overlap (i.e. do not share an internal point). More precisely, they are compatible if at least one of the following holds: $x_2\le x'_1$, $x'_2\le x_1$, $y_1\le y'_2$, or $y'_1\le y_2$.
\begin{definition}[associated rectangle]
With a task $i$ we associate the rectangle $(s_i,b(i),t_i,\ell(i))$.
\end{definition}

Note that $s_i$ and $t_i$ are integers since the path vertices are labeled $0,\ldots,m$, and that $b(i)$ and $\ell(i):=b(i)-d_i$ are non-negative integers as well. Tasks are called {\em compatible} if their associated rectangles are compatible, and a task set $F$ is called an {\em independent task set (\ISR)} if all tasks are pairwise compatible.

\begin{figure}[h]
\centering
\scalebox{1}{$\input{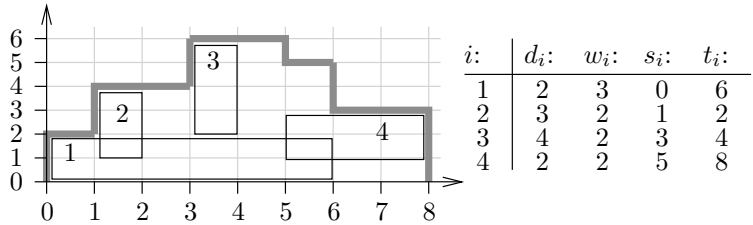}$}
\caption{A UFPP instance consisting of a path of length eight (with capacities 2,4,4,6,6,5,3,3) and four tasks, represented using rectangles drawn as high as possible under the capacity profile.}
\label{fig:TDset+instance}
\end{figure} 

The following geometric interpretation motivates this choice of rectangles, see Figure~\ref{fig:TDset+instance}. 
For every edge $e=\{x,x+1\}\in E(P)$, we draw a horizontal line segment between $(x,u_e)$ and $(x+1,u_e)$. We add vertical line segments to complete this into a closed curve from $(0,0)$ to $(m,0)$, called the {\em capacity profile}
(represented as a bold grey curve in Figure~\ref{fig:TDset+instance} and later figures).
For every task $i$, the associated rectangle is now a rectangle of height $d_i$, drawn as high as possible under the capacity profile.
The tasks in Figure~\ref{fig:TDset+instance} have profits 3, 2, 2, and 2, respectively. Therefore, an optimal UFPP solution consists of tasks 1, 3, and 4, with a total profit of 7. However, since the rectangles for the tasks 1 and 4 are incompatible, this is not an \ISR. The optimal \ISR\ consists of tasks 2, 3 and 4, with total profit 6. This \ISR\ is a UFPP solution as well. This is in fact always true, as shown in the next proposition.

\smallskip
\begin{proposition}\label{pro:TDset_feasible}
Let $F\subseteq T$ be an \ISR. Then $F$ is a feasible UFPP solution.
\end{proposition}

 \begin{proof}
Consider an edge $e\in E(P)$, and all tasks $T_e$ that use $e$.
Consider two tasks $i,j\in F\cap T_e$. They are compatible, but $s_i<e<t_j$ and $s_j<e<t_i$, so either $\ell(i)\ge b(j)$ or $\ell(j)\ge b(i)$ must hold.
It follows that the tasks in $F\cap T_e$ can be numbered $i_1,\ldots,i_p$ such that $b(i_1)\le \ell(i_2)$, $b(i_2)\le \ell(i_3)$, etc.
Hence,
\begin{eqnarray*}
\sum_{k=1}^{p}d_{i_{k}} & = & \sum_{k=1}^{p}(b(i_{k})-\ell(i_{k}))\\
 & = & b(i_{p})-\ell(i_{1})+\sum_{k=1}^{p-1}\underset{\le0}{\underbrace{(b(i_{k})-\ell(i_{k+1}))}}\\
 & \le & u_{e}. 
\end{eqnarray*}
 \end{proof}

In Section~\ref{sec:coloring}, we will first show that if all tasks are $\frac{1}{k}$-large, then
the profit of an optimal \ISR\ is at most a factor~$2k$ 
worse than the profit of an optimal UFPP solution. 
We will prove this by showing that any UFPP solution can be  partitioned into $2k$ \ISR s. 

After that we will give a dynamic programming algorithm for finding a maximum profit \ISR. The core concept is that of a {\em corner}, which corresponds to a region of the plane that contains a subset of the rectangles.
In Section~\ref{sec:partitioning} corners are defined, and a recursion formula is given for computing the profit of an optimal \ISR\ that fits in such a corner.
The dynamic programming algorithm that is based on this recursion is subsequently given in Section~\ref{sec:partitioning-dp}.

\subsection{Partitioning UFPP Solutions into \ISR s}\label{sec:coloring}

In this section we prove that if all tasks are $1/k$-large, the value of an optimal \ISR\ is at most a factor~$2k$ worse than 
the value of an optimal \UFPP\ solution, by showing that any \UFPP\ solution can be partitioned into $2k$ \ISR s.
This result is based on the following property.

\smallskip
\begin{proposition}
\label{propo:separatorbound}
Let $F$ be a feasible \UFPP\ solution, and let $k\ge 2$ be an integer, such that every task in $F$ is $1/k$-large. Let $e$ be a bottleneck edge for $i\in F$. Then there are at most $k-1$ tasks $j\in F\setminus \{i\}$ that are incompatible with $i$ and use $e$.
\end{proposition}
\begin{proof}
Let the set $F'\subseteq F\setminus \{i\}$ consist of all tasks that are incompatible with $i$ and that use $e$. Suppose $|F'|\ge k$.
Consider $j\in F'$. Since $j$ and $i$ are incompatible, but both use $e$ and $e$ is a bottleneck edge for $i$, it follows that $b(j)>\ell(i)=b(i)-d_i$.
Therefore,
$
d_i+\sum_{j\in F'} d_j \ge d_i + \frac{1}{k}\sum_{j\in F'} b(j) >
d_i+\frac{1}{k}|F'|(b(i)-d_i)\ge b(i) = u_e
$,
contradicting that $F$ is a feasible solution.
\qquad\end{proof}

A partition $\{F_1,\ldots,F_\ell\}$ of a task set $F$ into $\ell$ \ISR s will be encoded by an {\em $\ell$-coloring} $\alpha$ of $F$. This is a function $\alpha:F\rightarrow \{1,\ldots,\ell\}$ such that if $\alpha(i)=\alpha(j)$, then $i$ and $j$ are compatible. (So $\alpha(i)=j$ means that $i\in F_j$.)
An edge $e$ is called a {\em separator} for a task set $F$ if it is a bottleneck edge for some task $i\in F$, such that all tasks in $F$ that use $e$ are incompatible with $i$. So by Proposition~\ref{propo:separatorbound}, 
if $e$ is a separator edge, then there are at most $k-1$ tasks in $F$ that use $e$ and are incompatible with~$i$, 
when $F$ is a \UFPP\ solution consisting of $1/k$-large tasks.
A coloring $\alpha$ of $F$ is called {\em nice} if for every edge $e$ and task $i\in F$ that has $e$ as its bottleneck edge, all tasks that use $e$ and are incompatible with $i$ are colored differently.
The main idea behind these notions and our construction of a coloring is as follows: We will identify a separator edge $e$, and consider the set $F_0$ of tasks $i$ with $s_i<e$, and $F_1$ of tasks $i$ with $t_i>e$. (Note that $F_0\cup F_1=F$ and $F_0\cap F_1\not=\emptyset$.) Unless $F_0=F$ or $F_1=F$, we may use induction to conclude that both admit a nice $2k$-coloring. Then, since $e$ is a separator edge and the colorings are nice, all tasks in $F_0\cap F_1$ are colored differently in both colorings. Therefore these can be combined into a single nice $2k$-coloring for~$F$.
However, since it may occur that $F_0=F$ or $F_1=F$, we need a slightly more sophisticated argument, which we will now present in detail.

\smallskip
\begin{lemma}
\label{lem:coloring}
Let $F$ be a feasible \UFPP\ solution, and let $k\ge 2$ be an integer, such that every task in $F$ is $1/k$-large. Then there exists a nice $2k$-coloring for $F$. 
\end{lemma}
\begin{proof}
We prove the lemma by induction over~$|F|$.
The statement is trivially true when $|F|\le 2k$.
Now suppose that $|F|>2k$, and assume that the above statement holds for every such set $F'$ with $|F'|<|F|$. The proof is illustrated in Figure~\ref{fig:colorlem}.
Let~$e_B$ be a bottleneck edge for some task in $F$, 
with minimum capacity among all such edges.
Let $L\subseteq F$ be the set of tasks that use $e_B$. By Proposition~\ref{propo:separatorbound} and the choice of $e_B$, $|L|\le k$. 

Let $E_S=\{e_1,\ldots,e_p\}$ be the set of edges $e$ for which  
there is a task $i$ with $e$ as bottleneck edge, such that $i\in L$ or $i$ is incompatible with some $j\in L$. (So $e_B\in E_S$.) 
The edges in $E_S$ are numbered such that if $x<y$, then $e_x<e_y$.

We first argue that all edges in $E_S$ are separators for $F$.
For $e_B\in E_S$, the statement is clear since $e_B$ is a bottleneck edge with minimum capacity.
Now consider an edge $e\in E_S\setminus \{e_B\}$, and let $i\in F$ be a task with $e$ as bottleneck, that is incompatible with some task $i'\in L$. Then $\ell(i)<b(i')$ holds. Now suppose there is a task $i''$ using~$e$ that is compatible with $i$.
Since both $i$ and $i''$ use $e$ and $e$ is a bottleneck edge of $i$, from their compatibility it follows that $b(i'')\le \ell(i)<b(i')=e_B$. But this contradicts that $e_B$ is a bottleneck edge for $F$ with minimum capacity.

Now, for $1\le j\le p-1$, define $F_j\subseteq F$ to be all tasks $i$ with $s_i<e_{j+1}$ and $t_i>e_j$.
Similarly, define $F_0\subseteq F$ to be all tasks $i$ with $s_i<e_1$, and define $F_p\subseteq F$ to be all tasks $i$ with $t_i>e_p$.
Observe that $F_0\cup \ldots\cup F_p=F$.

\smallskip
\noindent
{\em Case 1:} For every $j$, $|F_j|<|F|$.

\smallskip
In this case, we may use induction to conclude that for every $F_j$ there exists a nice $2k$-coloring $\alpha_j$.
We can combine these into a nice $2k$-coloring of $F$: all tasks in $C=F_0\cap F_1$ use the edge $e_1$. Since $e_1$ is a separator, there exists a task $i\in C$ with $e_1$ as bottleneck such that all tasks using $e_1$ are incompatible with $i$.  
So, in both the nice $2k$-coloring $\alpha_0$ of $F_0$ and the nice $2k$-coloring $\alpha_1$ of $F_1$, all tasks in~$C$ are colored differently. Therefore, we can permute the colors of $\alpha_1$ such that the tasks in~$C$ receive the same colors in $\alpha_0$ and $\alpha_1$. At this point, the colorings can be combined into a $2k$-coloring $\alpha'$ of $F_0\cup F_1$, which is again nice.

Next, we can combine the nice $2k$-coloring $\alpha_2$ of $F_2$ with the nice $2k$-coloring $\alpha'$ of $F_0\cup F_1$ in a similar way, and continue like this with $F_3,\ldots,F_p$, until a nice $2k$-coloring of $F=F_0\cup \ldots\cup F_p$ is obtained. This proves the desired statement in this case.

\smallskip
\noindent
{\em Case 2:} There exists a $j$ such that $F_j=F$.

\smallskip
For such a choice of $j$, define $F'_j=F_j\setminus L$.
Since $L$ is non-empty, $|F'_j|<|F|$ holds, so we may again use induction to conclude that $F'_j$ admits a nice $2k$-coloring $\alpha$.
We will extend this to a nice $2k$-coloring of $F_j=F$.
For ease of exposition we will assume that $1\le j<p$ (the cases $j=0$ and $j=p$ are similar but easier).
So both $e_j$ and $e_{j+1}$ are edges in $E_S$.
We assume also w.l.o.g. that $e_B\ge e_{j+1}$.

We observe that every task $i\in F'_j$ that is incompatible with some task in $L$ uses either $e_j$ or $e_{j+1}$: by definition of $E_S$, such a task $i$ must use some edge $e_x\in E_S$. If $x>j+1$ then by definition of $F_j$, $i$ also uses $e_{j+1}$. If $x<j$ then similarly $i$ also uses~$e_j$.

Since $e_B\ge e_{j+1}$ and $F_j=F$, all tasks in $L$ use $e_{j+1}$. 
Then, since $e_{j+1}$ is a separator, there are at most $k-|L|$ tasks in $F'_j$ that use $e_{j+1}$ (Proposition~\ref{propo:separatorbound}); let $F_R$ denote these tasks. 
In addition, there are at most $k$ tasks in $F'_j$ that use $e_j$; denote these by $F_L$. 
Since $|F_L|+|F_R|\le 2k-|L|$ and there are $2k$ available colors, we can choose $|L|$ different colors that are not used for tasks in $F_L\cup F_R$, in the coloring $\alpha$. We extend the nice $2k$-coloring $\alpha$ for $F'_j$ to $F_j=F$ by using these $|L|$ colors for the tasks in $L$.
Since we observed that all tasks in $F'_j$ that are incompatible with some task in $L$ are included in $F_L\cup F_R$, it follows that this yields again a nice $2k$-coloring for $F$.
Since we now constructed a nice $2k$-coloring for $F$ in all cases, this concludes the proof of 
Lemma~\ref{lem:coloring}.
\qquad\end{proof}

\begin{figure}[htb]
\begin{center}
\scalebox{1}{\input{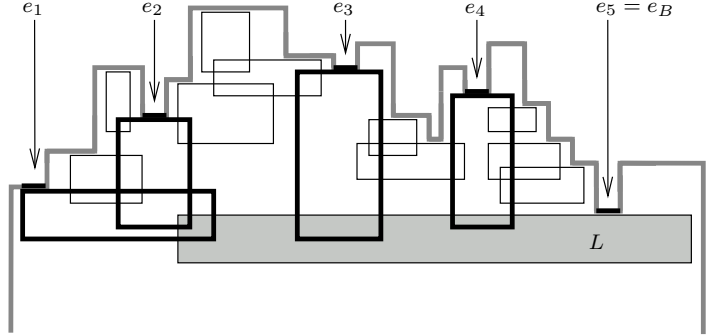}}
\end{center}
\caption{An illustration of the proof of Lemma~\ref{lem:coloring} for the case $k=2$ (using logarithmic scale on the $y$-axis, all tasks are $\frac{1}{2}$-large).
The gray task is the single task in $L$, and the bold tasks are associated with the indicated separator edges $e_1,\ldots,e_5$.}
\label{fig:colorlem}
\end{figure}

We will now show that this bound is tight for every $k\ge 2$: consider a path $P$ of length 5, with capacities $2k^2$, $2k^2+2k$, $2(2k^2+2k)$, $2k^2+2k$ and $2k^2$, in order along the path. (So $V(P)=\{0,\ldots,5\}$.)
Introduce $k-1$ tasks with $s_i=0$, $t_i=3$ and $d_i=2k+1$,
$k-1$ tasks with $s_i=2$, $t_i=5$ and $d_i=2k+1$,
one task with $s_i=1$, $t_i=3$ and $d_i=2k+3$, and 
one task with $s_i=2$, $t_i=4$ and $d_i=2k+3$.
All tasks can be verified to be $\frac{1}{k}$-large. 
In addition, they all satisfy  $b(i)\ge 2k^2>\ell(i)$ and $s_i\le 2<3\le t_i$, 
so all are pairwise incompatible, and therefore at least $2k$ colors are needed.
Finally, they constitute a feasible \UFPP\ solution: 
for the first and fifth edge this is easy to see. 
For the second and fourth edge, 
the sum of demands using the edge is $(k-1)(2k+1)+(2k+3)=2k^2+k+2\le 2k^2+2k=u_{\{1,2\}}=u_{\{3,4\}}$, since~$k\ge 2$.
All tasks use the third edge, so the demand sum is  
$2 \cdot((k-1)(2k+1)+(2k+3)) = 2\cdot (2k^2+k+2) = 4k^2 + 2k +4 \leq 4k^2 + 4k = u_{\{2,3\}}$, since~$k\geq 2$.
The above is summarized in the following proposition.

\begin{proposition}
 For any $k \in \mathbb{N}$ with $k\ge 2$ there is an instance of \UFPP\ with the properties that
 \begin{itemize}
  \item the instance consists of $2k$ tasks which are all $\frac{1}{k}$-large, 
  \item the set of all tasks of the instance yields a feasible \UFPP\ solution, and   
  \item the optimal \ISR\ contains only one task.
 \end{itemize}
\end{proposition}

\subsection{Computing Optimal \ISR s in Corners}
\label{sec:partitioning}

To bound the complexity of our algorithm, it is useful to assume that all edge capacities are different. 
In the next lemma we first show that by making small perturbations to the capacities, this property can easily be guaranteed, without changing the set of feasible solutions.

\smallskip
\begin{lemma} 
\label{lem:perturb-bcap}
Let $\mathcal{I}$ be a \UFPP\ instance with task set $T=\{1,\ldots,n\}$.
In linear time, the capacities and demands can be modified such that for the resulting instance~$\mathcal{I}'$:
\begin{itemize}
\item
All edge capacities are distinct, and
\item
a set of tasks $F\subseteq T$ is feasible for $\mathcal{I}$ if and only if $F$ is feasible for $\mathcal{I}'$.
\end{itemize}
\end{lemma}
\begin{proof}
Recall that the $m$ path edges are labeled $\{i,i+1\}$, for $i\in \{0,\ldots,m-1\}$.
For every edge $e=\{i,i+1\}$, we change the capacity to $u'_e=m u_e +i$.
For every task $j\in T$, we change the demand to $d'_j=m d_j$. This gives the instance $\mathcal{I}'$, in which all capacities are distinct.

Suppose $F\subseteq T$ is feasible for $\mathcal{I}'$. Then for every edge $e=\{i,i+1\}$, 
it holds that
$\sum_{i\in T_e\cap F} d_i = \frac{1}{m}\sum_{i\in T_e\cap F} d'_i  \le \lfloor\frac{1}{m}u'_e\rfloor = \lfloor u_e+\frac{i}{m}\rfloor=u_e$.
Hence $F$ is also feasible for the original instance $\mathcal{I}$.
Clearly, every task set $F$ that is feasible for $\mathcal{I}$ remains feasible for $\mathcal{I}'$.
\qquad\end{proof}

The central concept of our dynamic program is that of a {\em corner $(x,y,z)$}.
First we give an informal, geometric explanation of this notion, 
using the representation of the problem explained before. 
Subsequently, we will give a formal definition.

A corner $(x,y,z)$ corresponds to a certain region under the capacity profile, see Figure~\ref{fig:corner}.
In our dynamic program, we will compute for every such corner the profit of an optimal top-drawn solution that fits entirely within this region. This will be done using previously computed values for `smaller' corners. 
A corner $(x',y',z')$ is {\em smaller} than the corner $(x,y,z)$ if its region is a strict subset of the region associated with $(x,y,z)$ (we will make this more precise in the proof of Lemma~\ref{lem:terminates}).

\begin{figure}[htb]
\begin{center}
 \scalebox{0.6}{$\input{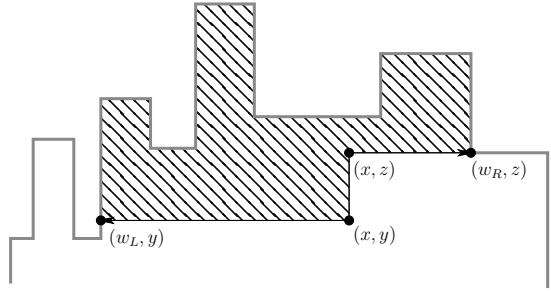}$}
\end{center}
\caption{The corner $\corner(x,y,z)$ contains all tasks for which the associated rectangles are fully contained in the shaded region.}
\label{fig:corner}
\end{figure}
A corner is determined by an integer $x$-coordinate $0\le x\le m$ (a path vertex), and two integer $y$-coordinates $y\ge 0$ and $z\ge 0$.
If the capacity of the edge $\{x-1,x\}\in E(P)$ is more than $y$, then draw a horizontal line segment from the point $(x,y)$ to the left, to the first point that lies on the capacity profile curve. This point will be called $(w_L,y)$. 
Otherwise, if the capacity is at most $y$, then let $w_L=x$.
Similarly, if the capacity of $\{x,x+1\}\in E(P)$ is more than $z$, then draw a horizontal line segment from $(x,z)$ to the right, to the first point that lies on the capacity profile. This point will be called $(w_R,z)$. Otherwise, let $w_R=x$.
Connect these line segments into a single curve from $(w_L,y)$ to $(w_R,z)$ by adding a vertical line segment from $(x,y)$ to $(x,z)$. 
This curve and the capacity profile together now enclose a bounded region, which is shown as a shaded region in Figure~\ref{fig:corner}.
(In the special case that $u_{\{x,x+1\}}>z\ge u_{\{x-1,x\}}>y$ or $u_{\{x-1,x\}}>y\ge u_{\{x,x+1\}}>z$, the corner actually corresponds to two disjoint regions.)
This is the region shown that we associate with the corner $(x,y,z)$; we say that a task {\em fits into} the corner $(x,y,z)$ when its rectangle, drawn as explained earlier, lies fully in (the closure of) this region. 
Now we give formal definitions.

\smallskip
\begin{definition}
A triple $(x,y,z)$ of integers is called a {\em \ct } if $0\le x\le m$, $y\ge 0 $ and $z\ge 0$.
For a \ct\ $(x,y,z)$, 
we denote by~$w_L(x,y)$ or simply $w_L$ the lowest numbered path vertex such that 
for all $i$ with $w_L\le i<x$, it holds that $u_{\{i,i+1\}} > y$.
Similarly,~$w_R(x,z)$ or simply $w_R$ is defined to be the largest numbered path vertex such that for all $i$ with $w_R\ge i>x$, it holds that $u_{\{i-1,i\}} > z$.
Hence, $w_L\le x$ and $w_R\ge x$.
\end{definition}

For each corner $(x,y,z)$ we define a set  $\corner(x,y,z)$ which---intuitively speaking---consists of all tasks for which the associated rectangles are contained in the region associated with this corner.

\smallskip
\begin{definition}
For a \ct\ $(x,y,z)$, 
we denote by $\corner(x,y,z)$ the set of all tasks $i\in T$ for which
at least one of the following holds:
\begin{itemize}
\item
$w_L(x,y)\leq s_i$,~$t_i\leq w_R(x,z)$ and~${\ell}(i) \geq \max\{y,z\}$,
\item
$w_L(x,y)\leq s_i$,~$t_i\leq x$ and~${\ell}(i) \geq y$, or
\item
$x\leq s_i$,~$t_i\leq w_R(x,z)$ and~${\ell}(i) \geq z$.
\end{itemize}
\end{definition}

Due to this definition, we say that {\em task $i$ fits into the corner $(x,y,z)$} or {\em corner $(x,y,z)$ contains $i$} if $i\in \corner(x,y,z)$. Hence, $\corner(x,y,z)$ is the set of all tasks that fit into the corner. Note that they are possibly incompatible.

For a given \UFPP\ instance and \ct\ $(x,y,z)$, we denote by $P(x,y,z)$ the maximum value of $w(F)$ over all \ISR s $F$ with $F\subseteq \corner(x,y,z)$. 
An \ISR\ $F$ with $F\subseteq \corner(x,y,z)$ and $w(F)=P(x,y,z)$ is said to {\em determine} $P(x,y,z)$.
Let $u_{\max}:=\max_{e\in E(P)} u_e$.
Observe that all tasks fit into the corner $(m,0,u_{\max})$, so computing $P(m,0,u_{\max})$ yields the desired value:

\smallskip
\begin{proposition}
\label{propo:opttopdrawn}
$P(m,0,u_{\max})$ equals the profit of an optimal \ISR, where $m$ is the path length.
\end{proposition}

We will now show how $P(x,y,z)$ can be computed in various cases.
Proposition~\ref{propo:easycases} collects the easy cases (which are easily understood using the above geometric interpretation),
and Lemma~\ref{lem:proper_recursion} covers the complex case which gives the main recursion formula
to compute an optimal \ISR. 
In Proposition~\ref{propo:easycases}, statement~(i) and~(ii) show how to rewrite a \ct\ in `standard form' without changing the corresponding region. Statement~(iii) deals with the case where the corner corresponds to two disconnected regions.

\smallskip
\begin{proposition} \hfill
\label{propo:easycases}
\begin{itemize}
\item[(i)] If $y=z$, then $\corner(x,y,z)=\corner(w_R(x,z),y,u_{\max})$, and hence~$P(x,y,z)=P(w_R(x,z),y,u_{\max})$.
\item[(ii)] If $x=0$ or $y\ge u_{\{x-1,x\}}$, then $\corner(x,y,z)=\corner(x,u_{\max},z)$, and hence~$P(x,y,z)=P(x,u_{\max},z)$.
\item[(iii)] If $y<z$, $x\ge 1$, $y<u_{\{x-1,x\}}$ and $u_{\{x-1,x\}}\le z<u_{\{x,x+1\}}$, then
  $\corner(x,y,z)=\corner(x,y,u_{\max})\cup\corner(x,u_{\max},z)$ and hence $P(x,y,z)=P(x,y,u_{\max})+P(x,u_{\max},z)$.
\end{itemize}
\end{proposition}
\begin{proof}
The statements (i) and (ii) follow immediately from the definitions. We will now prove statement (iii).
Consider an \ISR\  $F$ that determines $P(x,y,z)$.
Since $u_{\{x-1,x\}}\le z$, every task $i\in F$ with $s_i\le x-1$ and $t_i\ge x$ has $b(i)\le z$.
Hence $F$ contains no tasks~$i$ with $s_i\le x-1$ and $t_i\ge x+1$, and thus can be partitioned into two \ISR s that fit into the corners $(x,y,u_{\max})$ and $(x,u_{\max},z)$ respectively. 
It follows that $P(x,y,z)\le P(x,y,u_{\max})+P(x,u_{\max},z)$.

Now let $F_1$ and $F_2$ be \ISR s that determine $P(x,y,u_{\max})$ and $P(x,u_{\max},z)$. 
These are disjoint and both fit in the corner $(x,y,z)$, so $P(x,y,z)\ge w(F_1)+w(F_2)=P(x,y,u_{\max})+P(x,u_{\max},z)$.
\qquad\end{proof}

Because of Proposition~\ref{propo:easycases}~(i), 
we only need to consider corners $(x,y,z)$ where $y<z$ or $y>z$ holds.
Since these cases are symmetric, {\em we will only consider corners $(x,y,z)$ for which $y<z$ holds, in the following definitions and lemmas.}
One can easily deduce the corresponding statements for the case $y>z$. For the sake of brevity and readability, we will however not explicitly treat this case.
By  Proposition~\ref{propo:easycases}~(ii), we may furthermore assume that $x\ge 1$ and $y< u_{\{x-1,x\}}$.
Informally, Proposition~\ref{propo:easycases}~(iii) shows that we only need to consider corners for which the associated region is connected.
These observations show that we may restrict our attention to special types of corners, which we call \emph{proper}.

\smallskip
\begin{definition}
\label{def:proper-corner}
 A corner $(x,y,z)$ is called {\em proper} if $y<z$, $x\ge 1$,  $y<u_{\{x-1,x\}}$ and 
either $z<u_{\{x-1,x\}}$ or~$z\geq u_{\{x,x+1\}}$ holds.
\end{definition}

\begin{figure}[tb]
\begin{center}
\scalebox{1.15}{$\input{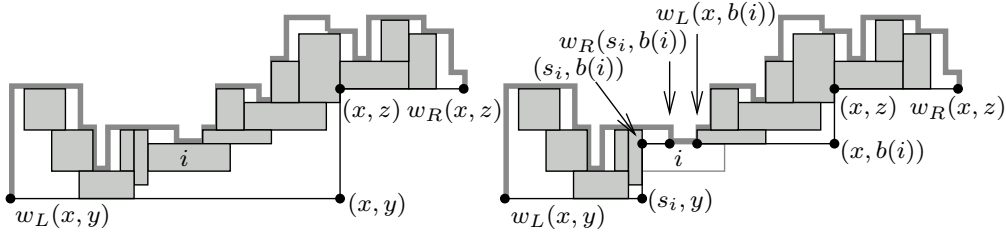}$}
\end{center}
\caption{On the left, a proper corner~$(x,y,z)$ and an \ISR\ $F\subseteq\corner(x,y,z)$ is shown. When choosing a special task $i\in F$, 
it holds that~$F\setminus \{i\} \subseteq \corner(s_i,y,b(i))\cup \corner(x,b(i),z)$.}
\label{fig:cornerpart_good}
\end{figure}
Now we will consider the more complex case where the given \ct\ $(x,y,z)$ is proper. The main idea is as follows.
Consider an \ISR\ $F$ that determines $P(x,y,z)$. Either $F$ also fits into the smaller corner $(x-1,y,z)$, or there exists a task $j$ with $\ell(j)<z$ and $t_j=x$. In the latter case, we show that $F$
can be partitioned into two task sets that fit into smaller corners, and one single task~$i$. 
Given a task $i$ with $t_i\le x$, we will consider the two smaller corners $(s_i,y,b(i))$ and $(x,b(i),z)$. 
These are illustrated in Figure~\ref{fig:cornerpart_good}.
For the indicated task $i$ and the depicted top-drawn set $F$, we have the desired property that $F\bs \{i\}\subset \corner(s_i,y,b(i))\cup \corner(x,b(i),z)$.
We will show that there always exists a task $i$ for which this holds; such a task is called {\em special}. 
This explains the following recursive formula for proper \ct s.

\smallskip
\begin{lemma}
\label{lem:proper_recursion}
Consider a \UFPP\ instance in which all edge capacities are distinct.
Let $(x,y,z)$ be a proper \ct.
Then 
\[
P(x,y,z) = \max \bigg\{ P(x-1,y,z), \max_{i\in \corner(x,y,z), t_i\le x}\Big\{w_i+P\big(s_i,y,b(i)\big)+ P\big(x,b(i),z\big) \Big\}  \bigg\}.
\]
\end{lemma}

\smallskip
Before we can prove Lemma~\ref{lem:proper_recursion}, we will define which properties are required for a special task, and prove that such a task always exists.
The essential property of a special task $i$ is that the rectangle $(s_i,b(i),x,y)$ 
is compatible with any associated rectangle of a task in~$F\setminus \{i\}$,
and does not overlap with the capacity profile.
In this case we will say that there are no tasks in $F$ that {\em lie to the right of $i$} or {\em lie below~$i$}, which is defined as follows.

\smallskip
\begin{definition}
Consider a proper \ct\  $(x,y,z)$, and a task $i\in \corner(x,y,z)$ with $t_i\le x$.
\begin{itemize}
\item
A task $j\in \corner(x,y,z)$ {\em lies to the right} of $i$ if $t_i\le s_j<x$ and $\ell(j)<b(i)$.
\item
A task $j\in \corner(x,y,z)$ {\em lies below} $i$ if $b(j)\le \ell(i)$ and $t_j>s_i$.
\end{itemize}
\end{definition}

\smallskip
\begin{definition}
Let $(x,y,z)$ be a proper corner and $F\subseteq \corner(x,y,z)$ be an \ISR.
A task $i\in F$ with~$t_i \leq x$ 
is called {\em special with respect to $F$ and $(x,y,z)$} 
if
\begin{itemize}
\item for all edges~$e$ with $s_i < e< x$, it holds that $u_e \geq b(i)$, and
\item there is no task $j\in F\setminus \{i\}$ that lies to the right of $i$ or that lies below~$i$.
\end{itemize}
\end{definition}

\smallskip
We now sketch how to find a special task~$i$ in case that~$F\nsubseteq \corner(x-1,y,z)$. 
We start with a task~$i$ such that~$t_i=x$ and~$\ell(i) < z$ holds. 
Now, assume there is a task~$j$ whose associated rectangle is incompatible with the rectangle~$(s_i,b(i),x,y)$.
Then,~$j$ lies below~$i$ (that is, $b(j)\le \ell(i)$), and no associated rectangle of a task
 crosses the line segment from $(s_j,b(j))$ to $(x,b(j))$.
This follows essentially from the fact that all tasks are compatible in $F$, and all associated rectangles touch the capacity profile.
We continue the same procedure with $j$ in the role of $i$ until we find 
a special task~$i$. 
We now prove this formally.

\smallskip
\begin{lemma}
\label{lem:emptyrectangle}
Let $(x,y,z)$ be a proper corner and $F\subseteq \corner(x,y,z)$ be an \ISR.
Then $F\subseteq \corner(x-1,y,z)$ or
there exists a special task $i\in F$ with respect to $F$ and $(x,y,z)$.
\end{lemma}
\begin{proof}
A task $i\in \corner(x,y,z)$ is called a {\em candidate} for $(x,y,z)$ if 
\begin{itemize}
\item
$t_i\le x$, 
\item
$\ell(i)<z$ and 
\item
for all edges $e\in E(P)$ with $s_i<e<x$, $u_e\ge b(i)$ holds.
\end{itemize}
From $u_{\{x-1,x\}}>y$ it follows that $w_L(x,y)=w_L(x-1,y)$.
Obviously, $w_R(x-1,z)\le w_R(x,z)$.
Therefore, since $y<z$ it holds that $\corner(x-1,y,z)\subseteq \corner(x,y,z)$.
Furthermore, $w_R(x-1,z)<w_R(x,z)$ can only occur if $u_{\{x-1,x\}}\le z$. But since the corner is proper, this implies that $u_{\{x,x+1\}}\le z$, so $w_R(x,z)=x$. 
These observations show that the only case when $F\subseteq \corner(x-1,y,z)$ does not hold is when there is a task $j\in F$ with $t_j=x$ and $\ell(j)<z$. Hence it is easily seen that $j$ is a candidate task (for $(x,y,z)$).
In addition, since $t_j=x$, no task in $F$ lies to the right of $j$.
So to prove the lemma statement, we may assume that there exists at least one candidate task in $F$ with no tasks to the right of it. Choose $i\in F$ to be such a task with minimum value for $b(i)$.
We prove that no task in $F$ lies below $i$, which will prove the lemma.
Suppose to the contrary that a task $j\in F$ lies below $i$, so $b(j)\le \ell(i)<b(i)$ and $t_{j}>s_i$. 
It is easily checked that $j$ is a candidate as well. 
By choice of $i$, it must then hold that
there exists a task $k\in F$ that lies to the right of $j$ (otherwise $j$ should have been chosen in the role of $i$).
So $t_{j}\le s_{k}< x$, and
$\ell(k)<b(j)\le \ell(i)<b(i)$.
Therefore, $s_{k}<t_i$; otherwise $k$ would also lie to the right of $i$. But now we can use the fact that $i$ and $k$ are compatible: Recall that we have that $s_{k}<t_i$, $s_i<t_{j}\le s_{k}$, and $\ell(k)<b(j)<b(i)$. So the only case in which $i$ and $k$ can be compatible is when $b(k)\le \ell(i)<b(i)$. This however contradicts the fact that all edges between $s_i$ and $x$ have capacity at least $b(i)$; the edge that determines $b(k)$ lies between $s_i$ and $x$.
We conclude that $i$ is a candidate task that satisfies the additional properties that there are no tasks in $F$ below it, or to the right of it, which therefore is special.
\qquad\end{proof}

Now we can prove Lemma~\ref{lem:proper_recursion}.

\smallskip
\begin{pfof}{Lemma~\ref{lem:proper_recursion}}
Let $(x,y,z)$ be a proper corner in a \UFPP\ instance in which all capacities are distinct. We will first show that 
\begin{equation}
\label{eq:ub}
P(x,y,z) \le \max \bigg\{ P(x-1,y,z), \max_{i\in \corner(x,y,z), t_i\le x}\Big\{w_i+P\big(s_i,y,b(i)\big)+ P\big(x,b(i),z\big) \Big\}  \bigg\}.
\end{equation}
Let $F$ be an \ISR\  that determines $P(x,y,z)$.
If $F\subseteq \corner(x-1,y,z)$ then $P(x,y,z)\le P(x-1,y,z)$.
If not, then by Lemma~\ref{lem:emptyrectangle}, there exists a special task~$i$.
We argue that $F\setminus \{i\}$ can be partitioned into two \ISR s $F_1$ and $F_2$ with $F_1\subseteq \corner\big(s_i,y,b(i)\big)$ and $F_2\subseteq \corner\big(x,b(i),z\big)$, which will prove that $P(x,y,z) \le \Big\{w_i+P\big(s_i,y,b(i)\big)+ P\big(x,b(i),z\big) \Big\}$, and therefore the Inequality~\eqref{eq:ub}.
First we consider the $w_L$ and $w_R$ vertices for both corners.
Let $e_B=\{v,v+1\}$ be a bottleneck edge for $i$, 
which is unique since all capacities are distinct. So 
$u_{e_B}=b(i)$ and $s_i<e_B<t_i$.
For the \ct\ $(s_i,y,b(i))$, it is easy to see that $w_L(s_i,y)=w_L(x,y)$ and $w_R(s_i,b(i))=v$.
Now consider the \ct\ $(x,b(i),z)$.
Obviously the $w_R$ value is the same as for $\corner(x,y,z)$.
In addition, since $i$ is special, all edges $e$ with $s_i<e_B<e<x$ have $u_e\ge b(i)$. Therefore, since we assumed all edge capacities are distinct, $u_e>b(i)$ holds for all such edges $e$. It follows that $w_L(x,b(i))=v+1$.

Now consider a task $j\in F\setminus \{i\}$. We distinguish five cases.
\begin{enumerate}
\item
$t_j\le s_i$: $w_L(x,y)=w_L(s_i,y)$ and $\ell(j)\ge y$ imply that $j\in \corner(s_i,y,b(i))$.
\item 
$s_i<t_j<e_B$: Since $j$ is compatible with $i$ but does not lie below $i$, $\ell(j)\ge b(i)$ holds. Therefore $j\in \corner(s_i,y,b(i))$.
\item
$s_j\ge x$: 
From $j\in\corner(x,y,z)$ it easily follows that $j\in\corner(x,b(i),z)$.
\item
$e_B<s_j<x$: Since $j$ is compatible with $i$ but does not lie to the right of~$i$, $\ell(j)\ge b(i)$ must hold. In addition, if $t_j>x$ then $\ell(j)\ge z$. Therefore $j\in \corner(x,b(i),z)$.
\item
$s_j<e_B$ and $e_B<t_j$: We argue that this is not possible. 
We have that $b(j)\le u_{e_B}=b(i)$. Because $i$ and $j$ are compatible, it then must be that $b(j)\le \ell(i)$. But this contradicts that there is no task below $i$.
\end{enumerate}
Since we considered all possibilities, this concludes the proof of Inequality~\eqref{eq:ub}.

\medskip
Next, we argue that
\begin{equation}
\label{eq:lb}
P(x,y,z) \ge \max \bigg\{ P(x-1,y,z), \max_{i\in \corner(x,y,z), t_i\le x}\Big\{w_i+P\big(s_i,y,b(i)\big)+ P\big(x,b(i),z\big) \Big\}  \bigg\}.
\end{equation}
Recall that since $(x,y,z)$ is a proper corner, 
$\corner(x-1,y,z)\subseteq \corner(x,y,z)$, so $P(x,y,z)\ge P(x-1,y,z)$.
Now consider a task $i\in \corner(x,y,z)$ with $t_i\le x$. 
Let $F_1$ and $F_2$ be the \ISR s that determine $P(s_i,y,b(i))$ and $P(x,b(i),z)$, respectively.
We will argue that $F_1\cap F_2=\emptyset$ and $F_1\cup F_2\subseteq \corner(x,y,z)$. 
Since clearly $i\not\in F_1$ and $i\not\in F_2$, it will follow that $P(x,y,z)\ge w_i+P(s_i,y,b(i))+ P(x,b(i),z)$, proving Inequality~\eqref{eq:lb}.

Now let $e_B$ a bottleneck edge for $i$.
Then $w_R(s_i,b(i))<e_B$.
On the other hand, since $t_i\le x$, it holds that $w_L(x,b(i))>e_B$. It follows that $F_1\cap F_2=\emptyset$.
Finally, we prove that both sets fit in the corner $(x,y,z)$.
Recall that $w_R(s_i,b(i))<e_B<t_i\le x$. 
Since in addition $b(i)>y$, it follows that $\corner(s_i,y,b(i))\subseteq \corner(x,y,z)$. (Because $w_R(s_i,b(i))<x$, it is irrelevant whether $b(i)>z$ or not.)
It is obvious that $\corner(x,b(i),z)\subseteq \corner(x,y,z)$, since $b(i)>y$.
This concludes the proof of Inequality~\eqref{eq:lb}.
\qquad\end{pfof}

\subsection{A Dynamic Programming Algorithm for Finding Optimal \ISR s}
\label{sec:partitioning-dp}

In this section we give an $O(n^4)$ dynamic programming algorithm for computing an optimal \ISR\ (Theorem~\ref{thm:MISR}). 
The main recursive routine $\PP(x,y,z)$ for computing $P(x,y,z)$ is shown in Algorithm~\ref{alg:DPstep}, which is based on the recursive formulas from the previous section (Proposition~\ref{propo:easycases} and Lemma~\ref{lem:proper_recursion}).
At the end of this section, we summarize how this algorithm, 
together with Lemma~\ref{lem:coloring}, yields the $2k$-approximation algorithm for UFPP in the case of~$1/k$-large tasks (Theorem~\ref{thm:largetasks}).

\begin{algorithm}
\label{alg:DPstep}
\caption{$\PP(x,y,z)$: a recursive algorithm for computing $P(x,y,z)$}

\Input{A \UFPP\ instance where $u_{\max}$ is the maximum edge capacity, all edge capacities are distinct, and integers $0\le x\le m$, $0\le y\le u_{\max}$, $0\le z\le u_{\max}$.}
\Output{$P(x,y,z)$.}
\BlankLine

\nl\lIf{$x=0$ or $y\ge u_{\{x-1,x\}}$}{$y:=u_{\max}$}\; \label{l:ytriv}
\nl\lIf{$x=m$ or $z\ge u_{\{x,x+1\}}$}{$z:=u_{\max}$}\; \label{l:ztriv} 
\nl
Compute $w_L:=w_L(x,y)$ and $w_R:=w_R(x,z)$\; \label{l:compute_ws}
\nl
\lIf{$w_L=w_R$}{\Return{0}}\;\label{l:emptycorner}
\nl
\lIf{$y=z$}{$z:=u_{\max}$; $x:=w_R$}\;\label{l:y_equals_z}
\nl
\If{$y<z$}{ \label{l:if_y_smaller}
  \nl \lIf{$u_{\{x-1,x\}} \le z<u_{\{x,x+1\} }$}{\Return($\PP(x,y,u_{\max})+\PP(x,u_{\max},z))$}\;\label{l:split}
    \nl \Return{$\max \Big\{ \PP(x-1,y,z),$\\\label{l:proper}
   \qquad \qquad \qquad $\max_{i\in \corner(x,y,z), t_i\le x}\Big(w_i+\PP\big(s_i,y,b(i)\big)+ \PP\big(x,b(i),z\big) \Big)  \Big\}$}\;
}
{\em (In the remaining case, $y>z$ holds:)}\;
  \nl
  \lIf{$u_{\{x,x+1\}}\le y<u_{\{x-1,x\}}$}{\Return($\PP(x,y,u_{\max})+\PP(x,u_{\max},z))$}\;\label{l:split2}
  \nl
  \Return{$\max \Big\{ \PP(x+1,y,z),$\\\label{l:proper2}
  \qquad\qquad \qquad $\max_{i\in \corner(x,y,z), s_i\ge x}\Big(w_i+\PP\big(t_i,b(i),z\big)+ \PP\big(x,y,b(i)\big) \Big)  \Big\}$}\;

\end{algorithm}

First, we show that the recursion terminates, by showing that all recursive calls are made 
with arguments $(x',y',z')$ that yield a `smaller' corner than the original argument $(x,y,z)$. 

\smallskip
\begin{lemma}
\label{lem:terminates}
Algorithm~\ref{alg:DPstep} terminates.
\end{lemma}
\begin{proof}
We show that when given an input $(x,y,z)$, all recursive calls will be with arguments $(x',y',z')$ such that the corner $(x',y',z')$ is strictly smaller than $(x,y,z)$, with respect to the following size measure:
For a corner $(x,y,z)$ with $y\le u_{\max}$ and $z\le u_{\max}$, we define
\[\area(x,y,z)=(x-w_L)(u_{\max}-y)+(w_R-x)(u_{\max}-z).
\] 
Note that for all corners $(x,y,z)$ that may be considered, $\area(x,y,z)$ is a non-negative integer. Hence the lemma statement then follows.

A recursive call is made in Line~\ref{l:split} whenever a corner $(x,y,z)$ is considered that satisfies the following bounds. Firstly $y<z\le u_{\max}$, and therefore $x\ge 1$ and $y<u_{\{x-1,x\}}$. Hence 
$w_L<x$. 
If the if-condition is satisfied then $w_R>x$ holds. 
This shows that for both calls $\PP(x,y,u_{\max})$ and $\PP(x,u_{\max},z)$, the $\area$ strictly decreases.
Now consider Line~\ref{l:proper}, which is reached whenever the corner is a proper corner.
In this case, when considering the corner $(x-1,y,z)$, the $w_L$ value does not change, and the $w_R$ value can only decrease.
Since $y<z$ it then follows that $\area(x-1,y,z)<\area(x,y,z)$.
Now consider a task $i\in \corner(x,y,z)$ with $t_i\le x$.
We have that $w_R(s_i,b(i))<x$ and $b(i)>y$, so $\area(s_i,y,b(i))<\area(x,y,z)$.
Clearly, $\area(x,b(i),z)<\area(x,y,z)$. So in all recursive calls in Line~\ref{l:proper}, the $\area$ decreases.
By symmetry, the same holds for Lines~\ref{l:split2} and~\ref{l:proper2}.
\qquad\end{proof}

\smallskip
\begin{lemma}
\label{lem:correct}
When given a corner $(x,y,z)$, Algorithm~\ref{alg:DPstep} returns $P(x,y,z)$.
\end{lemma}
\begin{proof}
We prove the statement by induction over the total number of recursive calls (which is possible since Lemma~\ref{lem:terminates} shows that the algorithm terminates).
That is, we will prove that a call $\PP(x,y,z)$ returns the value $P(x,y,z)$, assuming that this holds for recursive calls $\PP(x',y',z')$ that may be made.

The modification in Line~\ref{l:ytriv} yields an equivalent corner by Proposition~\ref{propo:easycases}(ii). 
By symmetry, the same holds for Line~\ref{l:ztriv}.
If $w_L=w_R$ then $\corner(x,y,z)=\emptyset$, so the correct answer is returned in Line~\ref{l:emptycorner}. 
The modification in Line~\ref{l:y_equals_z} is correct as well (Proposition~\ref{propo:easycases}(i)).
So we may assume now that $w_L<w_R$, and $y\not=z$.

Consider the case that $y<z$. This implies that $y<u_{\max}$, so $x\ge 1$ and $y<u_{\{x-1,x\}}$.
If $u_{\{x-1,x\}}\le z$ and $z<u_{\{x,x+1\}}$, then it is a non-proper corner, 
so Proposition~\ref{propo:easycases}~(iii) shows that Line~\ref{l:split} returns the correct answer.
So if Line~\ref{l:proper} is reached, the considered corner is proper.
By Lemma~\ref{lem:proper_recursion}, Line~\ref{l:proper} returns the correct answer.

The case that $y>z$ is analog; by symmetric arguments, the answers given in Lines~\ref{l:split2} and~\ref{l:proper2} are correct. This shows that Algorithm~\ref{alg:DPstep} returns the correct answer in every case.
\qquad\end{proof}

Lemma~\ref{lem:correct} and Proposition~\ref{propo:opttopdrawn} show that the profit of an optimal \ISR\ can be computed by calling Algorithm~\ref{alg:DPstep} with the argument $(m,0,u_{\max})$. However, to attain the desired complexity of $O(n^4)$, it is necessary to
use a dynamic programming table that stores previously computed values, to avoid repeating the same recursive calls. The complexity then follows from the fact that this table needs at most $n^3$ entries. 

\smallskip
\begin{theorem}
\label{thm:MISR}
There is an $O(n^4)$ algorithm for computing a maximum profit \ISR\ for a UFPP instance.
\end{theorem}
\begin{proof}
The algorithm is as follows.
First, in time $O(n+m)$, 
we transform the given instance into an equivalent instance, in which all edge capacities are distinct 
(Lemma~\ref{lem:perturb-bcap}).
As mentioned in Section~\ref{sec:prelim}, in time $O(n+m)$ we can transform the instance to an equivalent instance in which all vertices occur as the start or end vertex of some task, so we may furthermore assume that $m< 2n$. This preprocessing does not change $n$ and cannot increase $m$. 
On the resulting instance we call Algorithm~\ref{alg:DPstep}, with arguments $(x,y,z)=(m,0,u_{\max})$, to obtain the total profit of an optimal \ISR\  (Lemma~\ref{lem:correct}, Proposition~\ref{propo:opttopdrawn}). With a straightforward extension, we can compute an optimal \ISR\  itself.

A small modification is necessary to obtain a time complexity of $O(nm^3)$:
whenever any value $P(x,y,z)$ is computed during the course of computation, this value is stored in a table. Whenever a recursive call to $\PP(x,y,z)$ would be made, the algorithm instead returns the stored value $P(x,y,z)$, if it has been computed before (in a different recursion branch). This, together with Lemma~\ref{lem:terminates}, ensures that for every corner $(x,y,z)$ the routine $\PP(x,y,z)$ is called at most once during the course of computation. (Lemma~\ref{lem:terminates} ensures that a given corner is not considered more than once in the {\em same} recursion branch.)

We call a corner $(x,y,z)$ {\em relevant} if $y=0$ or $y$ is equal to the capacity of some edge, and if the same holds for $z$. 
Observe that when calling Algorithm~\ref{alg:DPstep} with a relevant \ct\ as argument, then every possible recursive call is with a relevant \ct\ as well.
So the total number of recursive calls is bounded by the total number of relevant \ct s. There are at most $m+1$ possibilities for $x$, and at most $m+1$ for $y$ and $z$ (since these need to correspond to actual capacities), so the total number of recursive calls is bounded by $O(m^3)$. 
Now consider the complexity of a single call. The computations in Line~\ref{l:compute_ws} can be done in time $O(m)$.
Ignoring recursive calls, Lines~\ref{l:split} and~\ref{l:split2} take constant time. 
Lines~\ref{l:proper} and~\ref{l:proper2} take time $O(n)$; for every task $i$, in constant time one can decide whether $i\in\corner(x,y,z)$ and $t_i\le x$ (resp.\ $s_i\ge x$), and evaluate the given expressions. 
The remaining lines clearly take constant time, so it follows that a single call of Algorithm~\ref{alg:DPstep} takes time $O(n)+O(m)\subseteq O(n)$. Hence, the total complexity becomes $O(nm^3)\subseteq O(n^4)$.
\qquad\end{proof}

Our results can be formulated as follows in terms of the Maximum Weight Independent Set of Rectangles (MWISR) problem:

\smallskip
\begin{theorem}
MWISR can be solved in time $O(n^4)$ for instances with $n$ axis-parallel rectangles, in which each rectangle contains a point $(x,y)$ such that no point $(x,y')$ with $y'>y$ is part of a rectangle.
\end{theorem}

We summarize the results of Section~\ref{sec:largetasks} in the following theorem.

\smallskip
\begin{theorem}
\label{thm:largetasks}
For every integer $k\ge 2$, there is an $O(n^4)$ time $2k$-approximation algorithm for \UFPP\ restricted to instances where every task is $\frac{1}{k}$-large.
\end{theorem}
\begin{proof}
In time $O(n^4)$ we compute an optimal \ISR\  $F$ for the instance (Theorem~\ref{thm:MISR}). This is a feasible UFPP solution (Proposition~\ref{pro:TDset_feasible}). Let $F^*$ be an optimal UFPP solution. Then $F^*$ can be partitioned into at most $2k$ \ISR s (Lemma~\ref{lem:coloring}), so $w(F)\ge \frac{1}{2k} w(F^*)$.
\qquad\end{proof}

\section{Main Approximation Algorithms}
\label{sec:mainapproxalgos}

We now apply the main results from the previous two sections to obtain our main approximation algorithm.

\smallskip
\begin{theorem}
\label{thm:7+eps-approx}For every $\epsilon >0$ there is a $(7+\epsilon)$-approximation
algorithm for the \UFPP\ problem. 
\end{theorem}
\begin{proof}
We partition the tasks into~$\frac{1}{2}$-small and~$\frac{1}{2}$-large tasks.
For the $\frac{1}{2}$-small tasks, we have a $(3+\epsilon)$-approximation algorithm (Theorem~\ref{thm:small-3+eps}), and for the $\frac{1}{2}$-large tasks, a $4$-approximation algorithm (Theorem~\ref{thm:largetasks}).
Returning the best solution of the two then yields a $(7+\epsilon)$-approximation algorithm (Fact~\ref{fact:additive_approximation}).
\qquad\end{proof}

The running time of the above algorithm is dominated by the dynamic
program which is needed for handling the `medium size' tasks, see Proposition~\ref{propo:DPpoly}.
Unfortunately, the exponent of the running time is quite large. However, in
the following theorem we show that already a moderate running time
of $O(n^{4}\log n)$ suffices to obtain a constant factor approximation, by partitioning the tasks differently in order to avoid this expensive dynamic program.
Hence, at the cost of a higher approximation factor, this algorithm
has a much better running time.

\smallskip
\begin{theorem}
\label{thm:constant-O(n^4)}
There exists a $25.12$-approximation algorithm
for \UFPP\ with running time $O(n^{4}\log n)$.
\end{theorem}
\begin{proof}
We partition the tasks into $\frac{1}{9}$-small and $\frac{1}{9}$-large tasks.
Using Theorem~\ref{thm:largetasks}, we obtain an $18$-approximation algorithm for the $\frac{1}{9}$-large tasks, that runs in time $O(n^4)$. 
For the $\frac{1}{9}$-small tasks, we apply our framework with constants $\ell=3$, $q=5$. According to Lemma~\ref{lem:framework}, we then choose $\beta := 2^{1-q} = 1/16$.
We first verify that these choices satisfy the conditions of Lemma~\ref{lem:LPtechnicalversion}:
\begin{itemize}
\item
$\delta \leq (1-\beta)/2^{\ell}$ holds, since~$(1-\beta)/2^{\ell}=15/128> 15/135 = 1/9 = \delta$. 
\item
$\delta/(1-\beta)\le (3-\sqrt{5})/2$ holds, since~$\delta/(1-\beta)=16/135\approx 0.12 < 0.38 \approx (3-\sqrt{5})/2$.
\end{itemize}
Hence we may apply Lemma~\ref{lem:LPtechnicalversion} to conclude that for every integer $k$, in time $O(n^3\log n)$, an $(f(\delta')/(1-\beta),\beta)$-approximative solution can be computed for the $\frac{1}{9}$-small tasks in $F^{k,\ell}$, where $\delta'=\delta/(1-\beta)=16/135$, and 
\[
f(\delta')=\frac{1+\sqrt{\delta'}}{1-\sqrt{\delta'}-\delta'}\approx \frac{1.3443}{0.5372}< 2.503.
\]
Therefore this is an $(\alpha,\beta)$-approximative solution, with $\alpha=2.503/(1-\beta)<2.67$.
Our framework (Lemma~\ref{lem:framework}) then gives an approximation algorithm for the instance consisting of all $\frac{1}{9}$-small tasks, with complexity $O(n^4\log n)$ and approximation ratio $\frac{\ell+q}{\ell}\cdot \alpha<7.12$.

We apply the $18$-approximation algorithm for the $\frac{1}{9}$-large tasks and the $7.12$-approximation algorithm for the $\frac{1}{9}$-small tasks. Returning the best solution of these two gives a $25.12$-approximation algorithm for the entire instance (Fact~\ref{fact:additive_approximation}).
\qquad\end{proof}

Chakrabarti et~al.~\cite{CCGK2007} showed that any $\alpha$-approximation algorithm 
for \UFPP\ implies a $(1 + \alpha + \epsilon)$-approximation algorithm for UFP on cycles (also called ring networks). 
Note that on a cycle, for each task $i$ two possible paths can be chosen from~$s_i$ to~$t_i$. (Even though in~\cite{CCGK2007} only UFP-NBA was considered, we observe that their argument applies to UFP as well.)
This yields the following corollary.

\smallskip
\begin{corollary}
\label{cor:8+eps-approx-ring}
For every $\epsilon >0$ there is a $(8+\epsilon)$-approximation
algorithm for the \emph{Unsplittable Flow Problem on Cycles}.
\end{corollary}

\section{Strong NP-Hardness\label{sec:hardnesssketch}}

\begin{figure} 
\centering
\scalebox{0.7}{$\input{constrExBWcompr.pstex_t}$}
\caption{An example of a graph $G$ with coloring $\alpha$ and the resulting instance \UFPPG. (The demand of a task is indicated by the height of its bar.)}
\label{fig:constrEx}
\end{figure} 

In this section we prove that \UFPP\ is strongly NP-hard for instances with demands in $\{1,2,3\}$, using a reduction from Maximum Independent Set in Cubic Graphs. 
Let $G,k$ be an independent set instance, where $G$ is a connected graph of maximum degree~3 with $G\not=K_4$. 
The question is whether $G$ contains an independent set of size at least $k$. 
This problem is NP-hard~\cite{GJS76} (even for {\em cubic} graphs, in which all degrees are exactly 3). Let $V(G)=\{v_1,\ldots,v_n\}$, and $E(G)=\{e_1,\ldots,e_m\}$.
By Brooks' Theorem (see e.g.~\cite{Lov75}), $G$ is 3-colorable since $G$ contains no $K_4$. Such a coloring can be found in polynomial time~\cite{Lov75}, so for our polynomial transformation we may assume that a proper 3-coloring $\alpha:V(G)\rightarrow \{1,2,3\}$ is given.

We now construct an equivalent instance of \UFPP\ as follows, see Figure~\ref{fig:constrEx}. 
The path $P$ has $2n+2m+1$ vertices, labeled $0,\ldots,2n+2m$.
For the proofs below it will be useful to distinguish different types of edges: 
an edge $\{i-1,i\}\in E(P)$ (with $1\le i\le 2n+2m$) is a {\em left edge} if $i\le 2m$ and a {\em right edge} otherwise. It is an {\em odd edge} if $i$ is odd, and an {\em even edge} otherwise.
Even left edges $\{2k-1,2k\}$ (with $1\le k\le m$) will be associated with edges $e_k$ of $G$, and even right edges $\{2k-1,2k\}$ (with $m+1\le k\le n+m$) will be associated with vertices $v_{k-m}$ of $G$.

For every vertex $v_i\in V(G)$, introduce the following tasks:
First, introduce one {\em long} task with start vertex $0$ and end vertex $2m+2i-1$.
Next, let $\sigma_1,\ldots,\sigma_d$ be the indices of the edges incident with $v_i$, in increasing order. Introduce $d+1$ {\em short} tasks with the following start and end vertex pairs:
$(0,2\sigma_1-1), (2\sigma_1,2\sigma_2-1),\ldots, (2\sigma_{d-1},2\sigma_d-1),
(2\sigma_d,2m+2i)$. 
For all of these tasks 
for vertex $v_i$, the demand is equal to $\alpha(v_i)$, and the profit is equal to $\alpha(v_i) n$ times the number of odd edges used by the task.
The aforementioned tasks will be called the {\em high-profit tasks} (for $v_i$).
Finally, for~$v_i$ we introduce one {\em low-profit task} from $2m+2i-1$ to $2m+2i$ with profit 1 and demand~$\alpha(v_i)$.

Doing this for all vertices gives the tasks of the \UFPP\ instance. It remains to set the edge capacities of $P$:
\begin{itemize}
\item
For odd left edges $e=\{2k-2,2k-1\}\in E(P)$ for integer $1\le k\le m$, set $u_e=\sum_{v\in V(G)} \alpha(v)$.
\item
For even left edges $e=\{2k-1,2k\}\in E(P)$ for integer $1\le k\le m$, set $u_e=(\sum_{v\in V(G)} \alpha(v))-1$.
\item
For right edges $e=\{2k-2,2k-1\}\in E(P)$ and $e=\{2k-1,2k\}\in E(P)$ for integer $m+1\le k\le n+m$, set $u_e=\sum_{i=k-m}^n \alpha(v_i)$.
\end{itemize}
This gives the instance \UFPPG\ of \UFPP, to which we will refer in the remainder of this section. With $T$ we will denote the set of all tasks of \UFPPG.
With an independent set $I$ of $G$ we associate a solution of \UFPPG\ of the following form.

\smallskip
\begin{definition}
\label{defi:stdform}
A set $F\subseteq T$ of tasks of \UFPPG\
is in {\em standard form} if:
\begin{enumerate}
\item
\label{stdform:prop_a}
For every $i\in \{1,\ldots,n\}$:
\begin{itemize}
\item
Either $F$ contains the long high-profit task for $v_i$ and the low-profit task for $v_i$, but no short high-profit tasks for $v_i$, 
\item
or $F$ contains all short high-profit tasks for $v_i$, but neither the long high-profit task for $v_i$ nor the low-profit task for $v_i$.
\end{itemize}
\item
\label{stdform:prop_b}
For every $\{v_i,v_j\}\in E(G)$, $F$ does not contain both the long high-profit task for $v_i$ and the long high-profit task for $v_j$.
\end{enumerate}
\end{definition}

\smallskip
Given an independent set $I$ of $G$, we can construct a solution $F$ to \UFPPG\ by selecting the long high-profit task and the low-profit task for $v_i$ if $v_i\in I$, and all short high-profit tasks for $v_i$ otherwise. 
Since $I$ is an independent set, this set $F$ is in standard form; it satisfies Property~\ref{stdform:prop_b} in the above definition.
We will first show that the resulting task set $F$ is a feasible solution to \UFPPG, by showing that in fact any task set in standard form is a feasible solution. Secondly, we show that the task set $F$ constructed this way is an optimal solution for \UFPPG.
To this end, we first argue that the total profit of a solution~$F$ in standard form depends only on the number of long high-profit tasks in~$F$. Subsequently, we will show that any optimal solution for \UFPPG\ is in standard form. 
Together this will show that the instances~$G$ and \UFPPG\ are equivalent (when considering them as decision problems with appropriately chosen target objective values).

\smallskip
\begin{proposition}
\label{propo:stdform_feasible}
Let $F\subseteq T$ be in standard form. Then $F$ is a feasible solution for \UFPPG.
\end{proposition}
\begin{proof}
Since $F$ is in standard form, for every $v_i\in V(G)$, every edge of $P$ is used by at most one task for $v_i$.
Therefore, for odd left edges $e=\{2k-2,2k-1\}\in E(P)$ (with $1\le k\le m$), the capacity $u_e=\sum_{v\in V(G)} \alpha(v)$ is not exceeded.
Similarly, for 
right edges $e=\{2k-2,2k-1\}\in E(P)$ (with $m+1\le k\le n+m$), the capacity $\sum_{i=k-m}^n \alpha(v_i)$ is not exceeded.
Now consider even left edges $e=\{2k-1,2k\}\in E(P)$ with $1\le k\le m$. These correspond to edges $e_k\in E(G)$. Let $e_k=\{v_i,v_j\}$. Since $F$ is in standard form, it cannot contain high-profit tasks for both $v_i$ and $v_j$. Note that no short high-profit tasks for $v_i$ and $v_j$ use the edge $\{2k-1,2k\}$. Hence the total capacity of $e$ used by $F$ is at most 
\[
\max\Big\{ \Big(\sum_{w\in V(G)} \alpha(w)\Big)-\alpha(v_i), \Big(\sum_{w\in V(G)} \alpha(w)\Big)-\alpha(v_j) \Big\}\le \Big(\sum_{w\in V(G)} \alpha(w)\Big)-1 = u_e.
\]
\end{proof}

\smallskip
\begin{proposition}
\label{propo:stdform_objval}
Let $F\subseteq T$ be a solution in standard form for \UFPPG, which contains $x$ long high-profit tasks. Then $w(F)= x+\sum_{i=1}^n \alpha(v_i)n(m+i)$.
\end{proposition}
\begin{proof}
Recall that for high-profit tasks for a vertex $v_i$, the profit equals $\alpha(v_i) n$ times the number of odd edges used by the task. The long-high profit task for $v_i$ (from $0$ to $2m+2i-1$) uses all odd edges $\{2k-2,2k-1\}$ of $P$ with $1\le k\le m+i$. Hence if selected, it contributes $\alpha(v_i)n(m+i)$ to the total profit. Together, the low-profit tasks for $v_i$ use exactly the same set of {\em odd} edges of $P$. So either way, the high-profit tasks for $v_i$ contribute $\alpha(v_i)n(m+i)$ to the total profit. The single low-profit task for $v_i$, with a profit of 1, is selected if and only if the long high-profit task is selected. This yields the stated total profit.
\qquad\end{proof}

\smallskip
\begin{lemma}
\label{lem:IStoFLOW}
If $G$ has an independent set $I$ with $|I|=x$, then \UFPPG\ admits a solution of profit $x+\sum_{i=1}^n \alpha(v_i)n(m+i)$.
\end{lemma}
\begin{proof}
Construct $F$ by selecting the long high-profit task and the low-profit task for $v_i$ if $v_i\in I$, and all short high-profit tasks for $v_i$ otherwise. Since $I$ is an independent set, this set $F$ is in standard form. So by Proposition~\ref{propo:stdform_feasible}, it is a feasible solution to \UFPPG. Proposition~\ref{propo:stdform_objval} now gives the total profit.
\qquad\end{proof}

Now we will show that any optimal solution to \UFPPG\ is in standard form, which will yield the converse of Lemma~\ref{lem:IStoFLOW}.
We say that a solution $F$ {\em fully uses the capacity} of an edge $e\in E(P)$ if $\sum_{i\in F\cap T_e} d_i=u_e$.
The majority of the total profit of a solution is proportional to the capacity use of all odd edges. Since there is at least one solution that fully uses the capacity of all of these edges (a solution in standard form), the next proposition follows.

\smallskip
\begin{proposition}
\label{propo:optsol_fullcap}
Let $F$ be an optimal solution for \UFPPG.
For every {\em odd} edge $e\in E(P)$, $F$ fully uses the capacity of $e$.
\end{proposition}
\begin{proof}
Recall that for every high-profit task with a demand of $d$, the profit equals $d n x$, where $x$ is the number of odd edges used by the task. 
Hence the total profit of high-profit tasks in a feasible solution $F$ is bounded by $n$ times the sum of capacities of all such edges. This capacity sum is 
\[
m\sum_{i=1}^n \alpha(v_i)+\sum_{k=1}^n \sum_{i=k}^n \alpha(v_i) =
\sum_{i=1}^n \alpha(v_i)(m+i).
\]
Suppose that for at least one odd edge the full capacity is not used by $F$. 
The low-profit tasks in $F$ can in total only yield a profit of $n$, so then the total profit is at most
\[
n\bigg(\sum_{i=1}^n \alpha(v_i)(m+i)-1\bigg)+n=n\sum_{i=1}^n \alpha(v_i)(m+i).
\]
Since there exists a feasible solution for \UFPPG\ with strictly higher profit (Lemma~\ref{lem:IStoFLOW}), this shows that $F$ is not optimal.
\qquad\end{proof}

Now we will apply Proposition~\ref{propo:optsol_fullcap} to prove that every optimal solution for \UFPPG\ is in standard form. 

\smallskip
\begin{proposition}
\label{propo:all_or_nothing}
Let $F$ be an optimal solution for \UFPPG.
Then for every $i\in\{1,\ldots,n\}$, $F$ either contains no short \hp\ tasks for $v_i$, or contains all short \hp\ tasks for $v_i$.
\end{proposition}
\begin{proof}
Suppose that the statement is not true. Then we may assume that there exists an even left edge $e=\{x,x+1\}$ in $P$ such that $F$ contains the short \hp\ task for $v_i$ that ends at $x$, but not the short \hp\ task for $v_i$ that starts at $x+1$ (the argument in the case where this is reversed is the same).
Let $v_iv_j\in E(G)$ be the edge that $e$ corresponds to, so $v_i$ and $v_j$ are the only vertices of $G$ for which $e$ is not included in the set of corresponding short \hp\ tasks.
Since $F$ fully uses the capacity of both of the odd edges $\{x-1,x\}$ and $\{x+1,x+2\}$ next to $e$ (Proposition~\ref{propo:optsol_fullcap}), $F$ must contain tasks starting on $x$ or $x+1$, with a total capacity of exactly $\alpha(v_i)$, to compensate for the fact that the short \hp\ task for $v_i$ starting at $x+1$ is not included. 
But the only such task is a short \hp\ task for $v_j$, which has a capacity of $\alpha(v_j)\not=\alpha(v_i)$ (since $\alpha$ is a proper coloring of $G$), a contradiction.
\qquad\end{proof}

\smallskip
\begin{lemma}
\label{lem:optsol_stdform}
Every optimal solution $F$ for \UFPPG\ is in standard form.
\end{lemma}
\begin{proof}
We first prove that Property~\ref{stdform:prop_a} of Definition~\ref{defi:stdform} holds. Considering the capacity $\alpha(v_n)$ of the rightmost odd edge $e=\{2m+2n-2,2m+2n-1\}$ shows that $F$ cannot contain both the long \hp\ task for $v_n$ and the short \hp\ task for $v_n$ that uses this edge. However, since $F$ fully uses the capacity of $e$ (Proposition~\ref{propo:optsol_fullcap}), exactly one of these cases is true. Combined with Proposition~\ref{propo:all_or_nothing} this shows that $F$ either contains the long \hp\ task for $v_n$, or all short \hp\ tasks. 
Since $F$ is optimal, the low-profit task for $v_n$ is included if and only if the long \hp\ task is included.
This shows that Property~\ref{stdform:prop_a} of Definition~\ref{defi:stdform} holds for $v_n$. Then 
the same argument can be applied using the odd edge $\{2m+2n-4,2m+2n-3\}$ to the left of $e$, which has capacity $\alpha(v_n)+\alpha(v_{n-1})$, to prove this property for $v_{n-1}$. Continuing this way proves Property~\ref{stdform:prop_a} of Definition~\ref{defi:stdform} for all $v_i$.

\medskip
Next, we will prove that Property~\ref{stdform:prop_b} of Definition~\ref{defi:stdform} also holds for the optimal solution $F$. We will prove that the vertices of $G$ for which $F$ contains a long high-profit task form an independent set in $G$.
More precisely, for any edge $\{v_i,v_j\}\in E(G)$, $F$ does not contain both the long high-profit task for $v_i$ and the long high-profit task for $v_j$.
Let $e_k = \{v_i,v_j\}$ (with $k\in \{1,\ldots,m\}$). According to our
construction this edge is associated with an even left edge $e=\{2k-1,2k\}\in E(P)$ of \UFPPG.
The capacity of $e$ is one unit smaller than the capacity of its adjacent odd edge $e'=\{2k-2,2k-1\}$, and yet the capacity of $e'$ is fully used by $F$ (Proposition~\ref{propo:optsol_fullcap}). The only tasks using $e'$ but not $e$ are a short task for $v_i$, and a short task for $v_j$. So for at least one of $v_i$ and $v_j$, $F$ includes this short task.
Since we already showed in the first part of the proof that $F$ cannot contain both a long high-profit task and a short high-profit task for the same vertex $v_k$, it follows that for at least one of $v_i$ and $v_j$, $F$ does not contain the long task. Hence Property~\ref{stdform:prop_b} of Definition~\ref{defi:stdform} also holds for $F$, and thus $F$ is in standard form.
\qquad\end{proof}

Combining the above lemma again with a calculation of the profit of a solution in standard form, this shows that any optimal solution to \UFPPG\ corresponds to an independent set of $G$:

\smallskip
\begin{lemma}
\label{lem:FLOWtoIS}
If \UFPPG\ admits a solution of profit at least $x+\sum_{i=1}^n \alpha(v_i)n(m+i)$, then $G$ has an independent set of size at least $x$.
\end{lemma}
\begin{proof}
Consider an optimal solution $F$, so $w(F)\ge x+\sum_{i=1}^n \alpha(v_i)n(m+i)$. 
By Lemma~\ref{lem:optsol_stdform}, $F$ is in standard form. 
Let $I$ be the set of vertices $v_i$ of $G$ for which $F$ contains the long high-profit task. 
Since $F$ is in standard form, $I$ is an independent set of $G$.
Proposition~\ref{propo:stdform_objval} shows that $w(F)= |I|+\sum_{i=1}^n \alpha(v_i)n(m+i)$. It follows that $|I|\ge x$, so $I$ is the desired independent set.
\qquad\end{proof}

Our construction uses only polynomially bounded numbers (profits, demands and capacities), so even if the numbers are encoded in unary, this is a polynomial transformation. For this it is also essential that a 3-coloring $\alpha$ of $G$ can be found in polynomial time~\cite{Lov75}. Lemmas~\ref{lem:IStoFLOW} and~\ref{lem:FLOWtoIS} show that $G$ has an independent set of size at least $x$ if and only if $P,F$ admits a solution of profit at least $x+\sum_{i=1}^n \alpha(v_i)n(m+i)$. Since the problem of finding a maximum independent set is NP-hard when restricted to graphs of maximum degree~3~\cite{GJS76}, this proves that UFPP\ is strongly NP-hard. The problem obviously lies in NP, which yields:

\smallskip
\begin{theorem}
UFPP is strongly NP-complete when restricted to instances with demands in $\{1,2,3\}$.
\end{theorem}

\smallskip
In fact, with a small addition we can show that the problem remains NP-complete 
when all edge capacities are equal.
Let $u_m$ be the maximum capacity used in the instance \UFPPG\ constructed above, 
and let $X=n+\sum_{i=1}^n \alpha(v_i)n(m+i)$ 
be an upper bound on the profit of any solution (Lemma~\ref{lem:FLOWtoIS}). 
Note that $X$ is bounded by a polynomial in $n$. 
For every edge $e=\{k,k+1\}$ of $P$ with $u_e<u_m$, 
we can increase the capacity to $u_m$, and introduce~$u_m-u_e$ {\em dummy tasks}, from $k$ to $k+1$ with demand~$1$ and profit~$X$.
For a polynomial transformation we must ensure that the number of dummy tasks is polynomially bounded. 
Given an instance with demands in~$\{1,2,3\}$, a maximum capacity of~$3n$ suffices, so w.l.o.g. $u_e\le 3n$ holds for all capacities.
Hence, on each edge at most~$O(n)$ dummy tasks are introduced. Since the number of edges
is bounded by~$O(n)$, at most~$O(n^2)$ dummy tasks with profit~$X$ 
and demand~$1$ need to be added. 
In an optimal solution, clearly all of the dummy tasks are selected, so 
this yields an equivalent instance. 
Therefore the previous theorem can be strengthened to: 

\smallskip
\begin{theorem}
\label{thm:NPC}
UFPP is strongly NP-complete when restricted to instances with demands in $\{1,2,3\}$, where all edges have equal capacities.
\end{theorem}

\section{Discussion}
\label{sec:discussion}

In this paper, we presented the first constant factor approximation algorithm for the unsplittable flow problem on a path. 
From a technical point of view, our key contributions are the framework employing the $(\alpha,\beta)$-approximation algorithms
for the small tasks, and the geometrically inspired dynamic program for the large tasks. 
In particular, for the latter we establish a connection between UFPP and Maximum Weight Independent Set of Rectangles, two problems which at first glance might seem totally unrelated. 
It may be interesting to study whether our new techniques for finding maximum non-overlapping sets of rectangles can be applied to other
geometric (packing) problems.

A challenging open question is whether a PTAS exists for UFPP, or whether the problem is APX-hard. 
Note that even in the much simpler case of uniform capacities (RAP), this question is open; the best known approximation algorithm has an approximation ratio of $2+\epsilon$~\cite{CCKR2002}.
It seems that for a PTAS one would have to use an approach that deals with all types of tasks (that is, small and large) in a combined manner,
rather than using a best-of-two analysis like in our algorithm, or in other algorithms such as the $(2+\epsilon)$-approximation algorithm for UFPP-NBA~\cite{I-MCF-factor4-trees}. 

We proved strong NP-hardness for UFPP with uniform capacities (RAP),
even if the input is restricted to demands in~$\{1,2,3\}$.
Due to this result it would be interesting to know whether the special case of demands in~$\{1,2\}$ 
is also NP-hard, or admits a polynomial time algorithm.
Note that if the demands are uniform, the problem can be solved in polynomial time with a small extension of the algorithm by Arkin and Silverberg~\cite{AS1987}.

Furthermore, our research raises the question what generalizations of UFPP also admit a constant factor approximation. One might consider generalizations such as those that have been considered for RAP and UFPP-NBA, either in a scheduling context or in a network flow context (see Section~\ref{ssec:related}).
In particular, an interesting open problem is whether the unsplittable flow problem on trees (without the NBA) also admits a constant factor approximation algorithm.

\subsection*{Acknowledgements}We would like to thank David Williamson for helpful comments on a draft of the paper, and Marek Chrobak for pointing us to~\cite{CWMX-ESA2010}.

\bibliographystyle{plain}
\bibliography{citations}

\appendix

\end{document}